\documentclass{amsart}

\usepackage{graphicx}
\usepackage{latexsym,color}
\usepackage{graphicx}
\usepackage{amssymb}
 \usepackage{amsfonts}
 \usepackage{amsmath}
\usepackage{amsthm}

\newtheorem{thm}{Theorem}[section]
 \newtheorem{dfn}[thm]{Definition}
 \newtheorem{lem}[thm]{Lemma}
 \newtheorem{rem}[thm]{Remark}

\newtheorem{asm}[thm]{Assumption}

\numberwithin{equation}{section}
\newcommand{\cA}{\mathcal{A}}

 \newcommand{\cD}{\mathcal{D}}

\newcommand{\cF}{\mathcal{F}}

\newcommand{\cL}{\mathcal{L}}

 \newcommand{\cQ}{\mathcal{Q}}

\newcommand{\cT}{\mathcal{T}}

\newcommand{\Om}{{\Omega}}

\def \D{\mathbb{D}}
\def \E{\mathbb{E}}
\def \F{\mathbb{F}}

 \def\P{\mathbb{P}}
\def \Q{\mathbb{Q}}
\def \R{\mathbb{R}}

\def \Sb{\mathbb {S}}

\def \cad{$c\grave{a}dl\grave{a}g$}

\def\reff#1{{\rm(\ref{#1})}}

\begin{document}
%\begin{frontmatter}
\title{Convex duality with transaction costs}

\thanks{Research of Dolinsky is partly supported by a
 European Union Career Integration grant CIG-618235
 and a Einstein Foundation Berlin grant A 20 12 137.
 Research of Soner
 is partly supported
 by the ETH Foundation, the Swiss Finance Institute
 and a Swiss National Foundation grant SNF 200021$\_$153555.}

  \author{Yan Dolinsky \address{
 Department of Statistics, Hebrew University of Jerusalem, Israel. \hspace{10pt}
 {e.mail: yan.dolinsky@mail.huji.ac.il}}}
\author{H.Mete  Soner \address{
 Department of Mathematics, ETH Zurich \&
 Swiss Finance Institute. \hspace{10pt}
 {e.mail: hmsoner@ethz.ch}}\\
  ${}$\\
 Hebrew University of Jerusalem
 and  ETH Zurich}

\date{\today}

\begin{abstract}
Convex duality for two two different  super--replication
problems in a continuous time financial market
with proportional transaction cost is proved.
In this market,
static hedging
in a finite number of options,
in addition to usual dynamic hedging
with the underlying stock, are allowed.
The first one of the problems considered is
the model--independent hedging that
requires the super--replication to hold for every
continuous path.  In the second one the market model is given
through a probability measure $\P$ and the inequalities are
understood $\P$ almost surely.
The main result, using the convex duality,
proves that the two super--replication problems have the same value
provided that $\P$ satisfies the conditional full support property.
Hence, the transaction costs prevents one from using the
structure of a specific model to reduce the
 super--replication cost.
 \end{abstract}

\subjclass[2010]{91G10, 60G44}
 \keywords{European Options, Model--free Hedging, Semi Static Hedging, Transaction Costs, Conditional Full Support }

\maketitle \markboth{Y.Dolinsky and H.M.Soner}{Duality with Transaction Costs}
\renewcommand{\theequation}{\arabic{section}.\arabic{equation}}
\pagenumbering{arabic}
%\end{frontmatter}

\section{Introduction}\label{sec:1}\setcounter{equation}{0}

The problem of super-replication is a convex optimization
problem in which the investor minimizes
the cost of a portfolio among those satisfying the
hedging constraints.
In the classical case the financial market
is frictionless and
the investors can buy or sell any quantity of
the stocks and other financial
instruments  at the same
price.  Then, the corresponding problem
is linear and
the optimization problem is in fact an infinite
dimensional
linear program.  In the
quantitative finance literature,
this problem is well studied and
is known to be related to arbitrage.
One central result is a convex duality
result, which contains deep
financial insights including the fundamental
theorem of asset pricing.

In the celebrated papers \cite{Dalang-Morton-Willinger, Dalbaen-Schachermayer, Kreps}
the financial market is modelled through a probability measure $\P$
that describes the future movements of the stock prices
in the time interval $[0,T]$.
The stock price process $S$
and the measure $\P$ are defined on a
probability space $\Om$ and a filtration $\F=\{\cF_t\}_{\{t \in [0,T]\}}$.
The main object of study
is an uncertain liability that will
be revealed in the future.  It is usually
modelled through  a $\cF_T$ measurable
random variable $\xi$
and the main goal is
to reduce the risk related to
$\xi$ by appropriately
trading in the financial market.
The investment opportunities
is given abstractly through a linear
set $\hat \cA$ denoting the
set of all admissible portfolios $\pi$
with a final portfolio value $Z^\pi_T$
at time $T$.
Then,
the super-replication problem
is to minimize the cost among all portfolios that
reduces the risk related to the liability $\xi$ to zero.  Mathematically,
\begin{equation}
\label{e.sp}
V(\xi):= \inf \ \left\{\
\cL(\pi)\ :\ \exists \pi \in \hat \cA\ {\mbox{such that}}
\ Z^\pi_T \ge \xi,\ \ \P-a.s.\ \right\},
\end{equation}
where $\cL(\pi) \in \R$ is the cost of
the portfolio $\pi$.
Once a market model is fixed through
a probability measure $\P$, then all statements are
supposed to be understood $\P$-almost surely.
Hence, the only role of the probability measure $\P$ is to describe the
null sets or equivalently
all impossible future scenarios.
Any other probability measure that
is equivalent to $\P$
(i.e., any measure with the same null sets)
would yield the same super-replication cost.
This problem is studied extensively
when the market is frictionless or equivalently $\cL$ is linear
and when only the adapted dynamic trading of the stock
without constraints is considered.
Under no-arbitrage type
assumptions and mild technical integrability
conditions, the convex dual is the following
maximization problem,
$$
D(\xi):= \sup_{\Q \in \cQ}\ \E_\Q\left[\xi \right],
$$
where $\cQ$ is the set of
all ``martingale'' measures that are equivalent to $\P$.
Precise statements in continuous time are technical
and we refer the reader to the seminal paper
of Delbaen \& Schachermayer \cite{Dalbaen-Schachermayer}.

These classical results were then extended
to markets with trading frictions.  It is shown that
super-replication in markets with
(proportional) transaction costs is prohibitively costly
as first proved in
Soner Shreve \& Cvitanic \cite{SSC}
and later generalized in Leventhal \& Skorohod \cite{LS},
Cvitanic, Pham \& Touzi \cite{CPT}, Bouchard \& Touzi \cite{BT},
Jakubenas, Levental, \& Ryznar \cite{JLR}, Guasoni, Rasonyi \& Schachermayer
\cite{GRS}, Blum \cite{B}
and for the game options in Dolinsky \cite{D}.
In all of these examples,
the super-replication cost is minimized
among all ``trivial'' strategies.  Hence,
the investor does not benefit from dynamic hedging
when the objective is to super replicate with certainty.
Also in all of these examples
not the null sets of $\P$ but rather the
support of it is important.
The related question of fundamental theorem of
asset pricing and super-heding duality with a given $\P$
is studied by Schachermayer \cite{S-FTAP,S} and the references therein.

One may reduce the hedging cost by including liquid derivatives in the super-replicating
portfolio. In particular, this might be the case
for semi-static hedging which is
detailed in the next section.  Namely, the investor
is allowed to take static positions in a
finite number of options (written on the underlying asset)
with  initially known prices.
In addition to these static option positions,
the stock is also traded dynamically and
all of these trades
are subject to proportional transaction
costs.
In terms of the above notation,
the set $\hat \cA$ of admissible
portfolios is enlarged by static option
trades but the transaction
costs make the cost functional $\cL$
to be convex rather than to be linear as in the classical papers.
We refer the reader to the survey of Hobson \cite{H},
a recent paper of the authors \cite{DS2}
and the references
therein for information
on semi-static hedging in continuous time.

While the model-independent approach with semi-static hedging
received considerable attention in recent years,
there are only few results for such markets with friction.
Indeed, recently the authors proved a model independent duality result
for semi--static hedging with transaction costs in discrete time \cite{DS1}.
Again in discrete time a fundamental theorem asset pricing
was  studied in Bayraktar \& Zhang \cite{BZ} and
in Bouchard \& Nutz \cite{BN}
in markets with transaction costs.
These later papers consider the quasi-sure criterion
given by a set of probabilistic models.
To the best of our knowledge, in  continuous time
semi--static hedging with transaction costs under model uncertainty
has not yet been studied.

In this paper, we consider a continuous time financial market
which consists of one risky asset
with continuous paths. In such a financial market
we study two super--replication problems of a given (path dependent)
European option.
We assume that the dynamic hedging of the stock
as well as the static option trading
are subject to transaction fees.
In the first problem, the market
model is given through a probability measure $\P$.
Then, the optimization problem corresponds to
 a straightforward extension
of \reff{e.sp}.
The second one is the model--independent
 problem
referring to super--replication  for all
continuous stock price processes.
Namely, in \reff{e.sp}
we require the inequality $Z^\pi_T \ge \xi$
to hold not $\P$-almost surely but rather
{\em{for every possible}} stock price path.
These definitions are given in  the subsection \ref{ss.spp} below.

Our main result Theorem \ref{thm.main}
states that these two problems
described above have the same value provided that
 the distribution
$\P$ of the stock price process satisfies the conditional full support property,
see Definition \ref{dfn.2} below.
Hence, in the presence of transaction costs the knowledge of
the model does not reduce the super--replication cost.
This explains the earlier results on super-replication with
friction and why the optimal hedge in these examples are the trivial ones.

Theorem \ref{thm.main} is proved under regularity Assumptions \ref{asm2.1},
\ref{asm.regularity} and a no-arbitrage type of condition
Assumption \ref{asm.noarbitrage}, below.
However, we do not assume any admissibility conditions on the portfolio.
Furthermore, we provide a duality result for the mutual
value in terms of consistent price systems
on the space of continuous functions that are
consistent with the option prices.  This duality
is very similar to the one proved in
discrete time in \cite{DS1}.

The proof of Theorem \ref{thm.main} is completed in four major steps.
First, we reduce the problem to bounded payoffs by applying the
pathwise inequalities which were obtained in
Acciaio {\em{et.al.~}} \cite{ABPST} and earlier by
Burkholder \cite{Bu}.
In the second step, we obtain a lower
bound
for the super--replication cost in the case where the model is given.
This bound is expressed in terms of modified model--free super--replication problems
 with appropriately lowered rate of transaction costs.
The third step is to derive an upper bound for the
model--free problem.
This step is done by
applying the recent results
of Schachermayer \cite{S} together with a lifting procedure
similar to the one developed in
our earlier work \cite{DS}.
The last step is a probabilistic proof
for the equality between (the asymptotic behaviour of) the lower and the upper bounds.

The paper is organised as follows. Main results are formulated in the next section.
In Section \ref{sec2+}, we reduce the problem to bounded claims.
A lower bound for the super--replication price in a given model is obtained
in Section \ref{sec3}.
Section \ref{sec4}  derives an upper bound for the model--free super--replication price.
The last section is devoted
to the proof of the equality
between the lower and the upper bounds.
\newpage

\section{Preliminaries and main results} \label{sec2} \setcounter{equation}{0}

\subsection{Market and Notation}
\label{ss.notation}

The  financial market consists of
 a savings account
 which is normalized to unity
 $B_t\equiv 1$
by discounting and of a risky asset $S_t$, $t\in [0,T]$, where $T<\infty$ is the maturity
date. Let  $s:=S_0>0$ be the initial stock price and without loss
of generality set $s=1$. We assume that the risky asset
could be any continuous process
with this initial data.

{In the sequel, we use the following notations.
For $s \ge 0$, $t \in [0,T)$, we set}
\begin{align*}
\mathcal{C}^+_s[t,T]&:= \left\{\ f:[t,T] \to [0,\infty)\ |\
f\ {\mbox{is continuous,}}\ f(t)=s\ \right\},\\
\mathcal{C}^{+}[t,T]&:=\bigcup_{s\geq 0} \mathcal{C}^{+}_s[t,T]
\end{align*}
and for $s>0$,
\begin{align*}
\mathcal{C}^{++}_s[t,T]
&:= \left\{\ f\in \mathcal{C}^+_s[t,T]\ |\
f(u) >0, \ \forall u \in [t,T]\ \right\},\\
\mathcal{C}^{++}[t,T]&:=\bigcup_{s>0} \mathcal{C}^{++}_s[t,T].
\end{align*}
Then,
$$
\Omega:= \mathcal{C}^{++}_1[0,T]
$$
represents the
set of all possible stock prices or the probability space.
We let
$\mathbb S=(\mathbb S_t)_{0\leq t\leq T}$ be the canonical
process given by $\mathbb S_t(\omega):=\omega_t$, for all $\omega\in\Omega$ and
 $\mathbb F_t:=\sigma(\mathbb
S_s,\, 0\leq s\leq t)$ be the canonical filtration (which is not right continuous).
We say that a probability measure $\Q$ on the space $(\Omega,{\mathbb F})$ is a martingale measure,
 if the canonical process $({\mathbb S_t})_{t=0}^T$ is a martingale with respect to $\Q$.

Further we  let
$$
\D[0,T]:=\left\{\ f:[t,T] \to [0,\infty)\ |\
f\ {\mbox{is \cad}}\ \right\},
$$
be the Skorokhod space of
\cad\ functions with the usual sup-norm
$$
\|\upsilon\|:=\sup_{0\leq t \leq
T}|\upsilon_t|.
$$

\subsection{The claim and its regularity}
\label{ss.regularity}

We model the liability of the claim through a deterministic
map of the whole stock price process.  Indeed,
for a given deterministic map
$$
G:\D[0,T]\rightarrow \mathbb{R}_{+},
$$
a general path dependent European option
has the payoff
$\xi=G(S)$.
Hence, although we consider only continuous stock price processes,
we implicitly assume that the option is defined for all
bounded measurable maps.

Our regularity assumption on the payoff is the same as
the one used in \cite{DS}.
For the convenience of the reader
we briefly review this assumption, but refer to \cite{DS} for
an extended discussion and its connection with the Skorokhod metric.
In particular, all options on the running maximum and Asian type options
satisfy it.
We make
the following standing assumption on $G$.

\begin{asm}
\label{asm2.1}
We assume that there exists a constant $L>0$ satisfying,\\ 
i.
$$
|G(\omega)-
G(\tilde\omega)|\leq L \|\omega-\tilde\omega\|,  \ \ \omega,\tilde\omega\in
\cD[0,T], 
$$
 ii.~and
 $$
|G(\upsilon)-G(\tilde\upsilon)| 
\leq L \|\upsilon\|\sum_{k=1}^{n}|\Delta t_k-\Delta \tilde t_k|, 
$$ 
for every piecewise constant functions
$\upsilon,\tilde\upsilon\in \cD[0,T]$
of the form
$$
\upsilon_t=\sum_{i=0}^{n-1} v_i \chi_{[t_i,t_{i+1})}(t)+v_n \chi_{[t_n,T]}(t) \ \mbox{and} \
\tilde\upsilon_t=\sum_{i=0}^{n-1} v_i \chi_{[\tilde t_i,\tilde t_{i+1})}(t)
+v_n \chi_{[\tilde t_n,T]}(t),
$$
where $t_0=0<t_1<...<t_n<T$, $\tilde t_0=0<\tilde t_1<...<\tilde t_n<T$ are two partitions
and as usual $\Delta t_k := t_k- t_{k-1}$, $\Delta
 \tilde t_k := \tilde t_k- \tilde t_{k-1}$, 
 $\chi_A$ is the characteristic function.
 \end{asm}
\vspace{8pt}

\subsection{Static Positions} \label{ss.call}
Next we describe the assumptions on the static options.
We assume that there are $N$ many options
$$
f_1,...,f_N:\D[0,T]
\rightarrow\mathbb{R}
$$
that are initially available for static hedging.
These options
may be path dependent.  We assume that their prices
$\cL_1,...,\cL_N\in \R $ are known and
that we can take {\em{static long positions}} on these options.
In this context, short positions can also be
allowed by including the negative of the options,
but the prices of these two (option and its negative) should add up to a positive
value equaling the bid-ask spread on this option.
Set
$$\cF(S):=(1,f_1(S),...,f_N(S)) \ \  \mbox{and} \ \
{\mathcal{L}}:=(1,\mathcal L_1,...,\mathcal L_N),
$$
where the first function which is identically equal
to one stands for investment in the non risky asset and we assume that
the investor can take long or short positions only in this option.  But
as discussed before, we allow only long positions in the
other options.
Thus, a static position in the
these options is represented by $c\in\mathbb R\times\mathbb{R}^N_{+}$
indicating an investment  of a European option
with the payoff $ c \cdot \cF(S)$
for the price
\begin{equation}
\label{e.l}
{\mathcal L}(c):=c \cdot \cL,
\end{equation}
where  $\cdot$ denotes
the standard inner product of $\mathbb{R}^{N+1}$.

We assume that the static options
satisfy some regularity assumptions and one of the static options
has a super quadratic growth.
More precisely, we assume the following.

\begin{asm}
\label{asm.regularity}
Functions $f_1,...,f_{N-1}$ satisfy Assumption
\ref{asm2.1}. We also assume that
if $f_i$ is path dependent
(i.e. do not depend only on the value of the stock at the maturity)
then it is bounded.
For $i=N$, we assume that
$f_N(\omega)=q(\omega_T)$ where
$q:\mathbb{R}_{+}\rightarrow\mathbb{R}_{+}$ is a convex function
satisfying
$$
|q(x)-q(y)|\leq L |x-y|\left(1+\frac{q(x)}{x}+\frac{q(y)}{y}\right), \ \ \forall x,y>0
$$
and
\begin{equation}
\label{e.cq}
\liminf_{x\rightarrow\infty}\ \frac{q(x)}{x^2}>0.
\end{equation}
\qed
\end{asm}
Since we consider hedging under proportional transaction costs,
 it is reasonable to assume that
the options $f_1(S),...,f_N(S)$ are also
subject to transaction costs. This together with no-arbitrage
considerations (see also \cite{BZ,BN})
leads us to the following assumption.

\begin{asm}
\label{asm.noarbitrage}
There exists a martingale
measure $\mathbb Q$ on the canonical
space $(\Omega,{\mathbb F})$ such that
$$
\mathbb{E}_{\mathbb Q}\left[f_i(\mathbb S)\right]<\mathcal L_i,
\ \ \forall i=1,\ldots N,
$$
where $\mathbb{E}_{\mathbb Q}$ denotes the expectation
with respect to the
probability measure $\mathbb Q$.

\qed
\end{asm}

\begin{rem}[{\bf{Comments on the assumptions}}]
\label{r.discussion}
{\rm{
{
In this paper we assume that there are only
finitely many static options. This setup is different from the one
in \cite{DS,DS1,DS2} where we assumed that the set of static options equals to
$\{f(S_T): f:\mathbb{R}_{+}\rightarrow\mathbb{R}\}$ (and includes power options).
The present assumptions seem to be more realistic.
We still assume that we have an option with super quadratic payoff $f_N$. This is needed for reducing the problem to bounded claims and
for dealing with
the hedging and the pricing error estimates arising in our discretization procedure.
In fact, it is sufficient to include an option with super linear payoff,
however for the simplicity of computations
we assume super--quadratic growth. }
Since the main focus of this paper is the equivalence
between two different super-replication problems, we do not
seek the most general assumptions on the static options.
It is plausible that the main result holds under weaker assumptions.
In particular, {for bounded claims one might be able to avoid the use
of the quadratic option as in \cite{DS1}}.

The second assumption states that
there exist a linear pricing rule
that is consistent with the observed option data.
This implies in particular no-arbitrage in this market.
Also the strict inequality implies that
the options are subject to proportional transaction costs.
The equivalence of no-arbitrage and the existence of such measures
is in fact a
difficult question.  Only recently  several discrete time results
in this direction were proved
in \cite{BZ,BN}.

\qed

}}
\end{rem}
\subsection{Hedging with transaction costs}
\label{ss.admissible}

We continue by describing the continuous time trading with
proportional transaction costs,
in the underlying asset $S$.
Let $\kappa\in (0,1)$ be
the proportional transaction cost rate.
Denote by $\gamma_t$ the number of shares of the risky asset in the portfolio $\pi$
at moment of time $t$ before the transaction at this time.
Due to transaction costs it has to be
of bounded variation.
Hence, we assume that the process
$\gamma={\{\gamma_t\}}_{t=0}^T$ is an adapted process (to the
raw filtration generated
by the stock price process)
of bounded variation with left continuous paths with $\gamma_0=0$.
Let
\begin{equation*}
\gamma_t=\gamma^{+}_t-\gamma^{-}_t
\end{equation*}
be a decomposition of $\gamma$ into positive
and negative variations.
Namely,
$\gamma^{+}_t$
denotes the cumulative number of stocks
purchased up to time $t$
not including the transfers made at time $t$
and respectively, $\gamma^{-}_t$,
denotes the cumulative number of stocks,
 sold up to time $t$ again
 not including the transfers made at time $t$.  Let $\cA$ be the set of
 all such processes.

In this financial market,
a {\em{hedge}} is a pair $\pi=(c,\gamma)\in \hat \cA:=\mathbb R\times\mathbb{R}^N_{+} \times \cA$
and the corresponding {\em{portfolio liquidation value}} at the maturity
date $T$ is given by
%\begin{eqnarray}
%\label{e.psp}
\begin{eqnarray*}
Z^{\pi}_T(S)& := &c \cdot {\cF}(S)+ \left[\gamma_T -\kappa |\gamma_T|\right]\ S_T\\
&& +(1-\kappa) \int_{[0,T]}{S}_u \ d\gamma^{-}_u-
(1+\kappa)\int_{[0,T]}{S}_u \ d\gamma^{+}_u,\nonumber
\end{eqnarray*}
%\end{eqnarray}
where the above integrals are the standard
Stieltjes integrals and ${\cF}(S)$ is as in subsection \ref{ss.call}.
Notice that the term $-\kappa |\gamma_T|\ S_T$ in the first line
is due to liquidation cost at maturity.
The {\em{cost of this portfolio}} $\pi=(c,\gamma)$ is equal to ${\mathcal L}(c)$
as defined in \reff{e.l}.

\subsection{Super--replication problems}
\label{ss.spp}

In this subsection, we
introduce two super--replication
problems. For the liability
$\xi=G(S)$, the model--free super--replication cost is defined by
%\begin{equation}
%\label{2.3}
$$
V_{\kappa}(G):=
\inf\left\{\cL(c): \exists \pi\in \hat \cA=\mathbb R\times\mathbb{R}^N_{+} \times \cA  \ {\mbox{so that}} \ \
Z^{\pi}_T(S)\geq G(S) \ \ \forall S\in \Omega\right\}.
$$
%\end{equation}
For the second problem, we assume that
a probability measure $\P$ on the canonical space $\Omega$ is given.
Then, the corresponding problem is
%\begin{equation}
%\label{e.msp}
$$
V^{\mathbb P}_{\kappa}(G):=\inf\left\{\cL(c): \exists \pi\
\in \hat \cA=\mathbb R\times\mathbb{R}^N_{+} \times \cA  \ {\mbox{so that}} \ \
Z^{\pi}_T(S)\geq G(S) \ \ \mathbb P-\mbox{a.s}\right\}.
$$
%\end{equation}

The main goal of this paper is to obtain the
convex duality for these functionals and prove that
they are equal if the measure $\P$ has
conditional full support as defined in the
next subsection.

\subsection{Main Results}
In order to formulate our results we need the following definitions.  Recall that
$\mathcal{C}^{++}[t,T]$ and
the canonical space $\Omega= \mathcal{C}^{++}_1[t,T]$ are
defined in subsection \ref{ss.notation}.

\begin{dfn}
\label{dfn.1}
{\rm{Consider the sample space
$\hat\Omega:=\Omega\times \mathcal{C}^{++}[0,T]$.
Let $\hat{\mathbb S}=(\mathbb S^{(1)},\mathbb{S}^{(2)})$ be the canonical process
on $\hat\Omega$ and
$\hat{\mathbb F}_t:=\sigma(\hat{\mathbb S}_s,\, 0\leq s\leq t)$
be the canonical filtration.
A $(\kappa,{\mathcal L})$}} consistent price system
{\rm{is a probability measure $\hat{\mathbb Q}$ on  $\hat\Omega$
satisfying,}}
\begin{enumerate}
\item $\mathbb S^{(2)}$ {\rm{is a}} $\hat{\mathbb Q}$
{\rm{martingale with respect to}} $\hat{\mathbb F}$;
\item
$(1-\kappa)\mathbb{S}^{(1)}_t\leq \mathbb S^{(2)}_t\leq (1+\kappa)\mathbb S^{(1)}_t$,
$\ \hat{\mathbb Q}${\rm{--a.s.}}
\item
$\mathbb{E}_{\hat{\mathbb Q}}\left[f_i(\mathbb S^{(1)})\right]\leq\mathcal L_i$,
\ {\rm{for all}} $i=1,\ldots,N.$
\end{enumerate}
{\rm{The set of all $(\kappa, {\mathcal L})$ consistent price systems is denoted by}}
$\mathcal{M}_{\kappa,{\mathcal L}}$.
\qed
\end{dfn}

Next we recall the notion of conditional full support.
As usual, the support of a a probability measure $\mathbb{P}$
on a separable space,
denoted by $supp\ \mathbb{P}$, is defined as the minimal closed set of full measure.

\begin{dfn}
\label{dfn.2}
{\rm{We say that a probability measure $\mathbb P$ has
the}} conditional full support property
{\rm{if  for all $t\in [0,T)$
%\begin{equation}
%\label{e.full}
 $$
 supp \ \mathbb P(\mathbb{S}_{|[t,T]}|\mathbb{F}_t)=\mathcal C^{+}_{{\mathbb S}_t}[t,T] \ \ \mbox{a.s.}
$$
%\end{equation}
where  $\mathbb P(\mathbb{S}_{|[t,T]}|\mathbb{F}_t)$
denotes the $\mathbb{F}_t$--conditional
distribution of the $\mathcal {C}^{+}[t,T]$
valued
random variable $\mathbb{S}_{|[t,T]}$
which is the restriction of the canonical process to
$[t,T]$.}}
\end{dfn}

We are ready to state our main result.

\begin{thm}\label{thm.main}
Suppose Assumptions \ref{asm2.1}, \ref{asm.regularity},
\ref{asm.noarbitrage} hold.
Assume $0<\kappa<1/8$ and let $\mathbb P$ be a probability measure which satisfies
the conditional full support property.
Then,
$$
V^{\mathbb P}_{\kappa}(G)=V_{\kappa}(G)=\sup_{\hat{\mathbb Q}\in
\mathcal{M}_{\kappa,{\mathcal L}}}
\mathbb E_{\hat{\mathbb Q}}[G(\mathbb S^{(1)})].
$$
\end{thm}
Clearly,  $V^{\mathbb P}_{\kappa}(G)\leq V_{\kappa}(G)$.
Therefore, in order
to prove Theorem \ref{thm.main} it suffices to prove
the following two inequalities,
\begin{equation}
\label{e.lower}
V^{\mathbb P}_{\kappa}(G)\geq\sup_{\hat{\mathbb Q}\in
 \mathcal{M}_{\kappa,{\mathcal L}}}\mathbb E_{\hat{\mathbb Q}}[G(\mathbb S^{(1)})]
\end{equation}
and
\begin{equation}\label{e.upper}
V_{\kappa}(G)\leq\sup_{\hat{\mathbb Q}\in
\mathcal{M}_{\kappa,{\mathcal L}}}\mathbb E_{\hat{\mathbb Q}}[G(\mathbb S^{(1)})].
\end{equation}
The lower bound \reff{e.lower} is proved in Lemma \ref{lem6.2}
and the upper bound \reff{e.upper} is established in Lemma \ref{lem6.3}.

In the sequel, we always assume, without explicitly stating, that $0<\kappa<1/8$.

\section{Reduction to Bounded Claims}
\label{sec2+} \setcounter{equation}{0}

The following result shows that in this market
one can hedge certain claims in the tails with small
cost. {Similarly, to \cite{DS,DS1}, the proof is done by combining assumption (2.2)
and the results of \cite{ABPST}.}

\begin{lem}\label{lem2+.1}
For any $K>0$
consider the European claim
$$
\alpha_K(S):=\frac{||S||}{K}+||S||\chi_{\{||S||\geq K\}}(S), \ \ S\in \Omega,
$$
where as before $\chi_A$ is the characteristic function.
Under Assumption \ref{asm.regularity},
$$
\lim_{K\rightarrow\infty}V_{\kappa}(\alpha_K)=0.
$$
\end{lem}
\begin{proof}
Let
$$
\theta_0:=\theta_0(S)=0
$$
and for a positive integer $k$ we recursively define the stopping times by,
$$
\theta_k:=\theta_k(S)=T\wedge\inf\{t>\theta_{k-1}: |S_t-S_{\theta_{k-1}}|=1\}.
$$
Let
$\mathbb K:=\mathbb K(S)=\min\{k:\theta_k=T\}$. Clearly, $\mathbb K<\infty$ for every $S\in\Omega$.
By \reff{e.cq},
it follows that there exists $c_q>1$ such that
\begin{equation}
\label{2.100}
q(x)\geq\frac{x^2}{c_q}, \ \ \forall x\ge c_q.
\end{equation}
Consider the portfolio $\pi=(c,\gamma)$ where\
$$
\gamma_t=-\sum_{i=0}^{\mathbb K-1}
\max_{0\leq j\leq i} S_{\theta_j}\
\chi_{(\theta_i,\theta_{i+1}]}(t) , \ \qquad t\in [0,T],
$$
and
$$
c=(c_q^2,0,...,0, c_q),
$$
i.e., we buy $c_q$ many options $q(S_T)$ and invest in the riskless asset
$ c_q^2$ dollars.
By summation by parts, Proposition 2.1 in Acciaio {\em{et.al}} \cite{ABPST}
(see also Burkholder \cite{Bu})
and  \reff{2.100}, it follows that
\begin{eqnarray*}
Z^\pi_T(S)&=& c_q^2+c_q q(S_T)-
\sum_{i=0}^{\mathbb K-1}\left[\max_{0\leq j\leq i} S_{\theta_j}\right] \ (S_{\theta_{i+1}}-S_{\theta_i})\\
&&-\kappa\ \sum_{i=1}^{\mathbb K-1}S_{\theta_i}\left[
\max_{0\leq j\leq i} S_{\theta_j}-
\max_{0\leq j\leq i-1} S_{\theta_j}\right]\\
&&- \kappa S^2_0\ -\ \kappa S_T\ \left[\max_{0\leq j\leq \mathbb K-1} S_{\theta_j}\right]\\
&\geq & \frac{(1-8\kappa)}{4}\ \max_{0\leq j\leq \mathbb K} S^2_{\theta_j}.
\end{eqnarray*}
Observe that
$$
||S||\leq 1+\max_{0\leq j\leq \mathbb K} S_{\theta_j}\leq 2\max_{0\leq j\leq\mathbb K} S_{\theta_j}.
$$
Also, since for any $S \in \Omega$,
$S_0=1$, $\|S\|\ge 1$.  Hence,
$$
K\alpha_K(S) \le \|S\|+ \|S\|^2 \le 2 \|S\|^2 \le 8
\max_{0\leq j\leq \mathbb K} S_{\theta_j}^2 .
$$
Thus, {(recall that $\kappa<\frac{1}{8}$)}
$$
Z^\pi_T(S)\geq \frac{(1-8\kappa)}{4}\
\left(\max_{0\leq j\leq\mathbb K} S^2_{\theta_j}\right)
\geq \frac{K(1-8\kappa)}{32}\alpha_K(S).
$$
We conclude that the super-replication cost of $[K(1-8\kappa)/32] \ \alpha_K$ is no more than
the cost of this portfolio.  Hence,
\begin{equation}
\label{2.200}
V_{\kappa}(\alpha_K)
\leq \frac{32}{(1-8\kappa)}\ \frac{c_q^2+c_q\mathcal{L}_N}{K}
\end{equation}
and the result follows after taking $K$
to infinity.
\end{proof}

Next, we establish the reduction to bounded claims.

\begin{lem}\label{lem2+.2}
Under the assumptions of Theorem \ref{thm.main},,
it sufficient to prove Theorem \ref{thm.main} for bounded
claims.
\end{lem}
\begin{proof}
Let $L$ be the Lispschitz constant in
Assumption \ref{asm2.1}. For any $K\ge 1$ set
$$
G_K(S):=G(S)\wedge [LK+G(0)],
\qquad
S\in\Omega.
$$
From Assumption \ref{asm2.1}, it follows that $G(S)\leq G(0)+L\|S\|$.
Therefore, for all  $K \ge 1$,
$$
G(S)\leq G_K(S)+(G(0)+L)\alpha_{K}(S).
$$
Consequently,
$$
V_{\kappa}(G)\leq V_{\kappa}(G_K)+(G(0)+L)V_{\kappa}(\alpha_{K}),
\quad
V^{\mathbb P}_{\kappa}(G)\leq V^{\mathbb P}_{\kappa}(G_K)+(G(0)+L)V_{\kappa}(\alpha_{K}).
$$
Since $G_K$ is bounded, if
Theorem \ref{thm.main} holds for such claim,
by the monotone convergence theorem we would have
\begin{equation*}
V_{\kappa}(G)=\lim_{K\rightarrow\infty}V_{\kappa}(G_K)
=\lim_{K\rightarrow\infty}
\sup_{\mathbb Q\in \mathcal{M}_{\kappa,{\mathcal L}}}\E_{\mathbb Q}[G_K(\mathbb S^{(1)})]=
\sup_{\mathbb Q\in \mathcal{M}_{\kappa,{\mathcal L}}}\E_{\mathbb Q}[G(\mathbb S^{(1)})].
\end{equation*}
Similar identities hold for $V^{\mathbb P}_{\kappa}(G)$
as well, proving the main theorem for all claims
satisfying the Assumption \ref{asm2.1}.
\end{proof}
From now on, we will assume (without loss of generality)
that there exists a constant $K>0$ such that
$0\leq G\leq K$.

\section{Lower Bound}
\label{sec3}
\setcounter{equation}{0}
In this section we establish estimates for the lower bound
\reff{e.lower}, under the assumptions of Theorem \ref{thm.main}. We start with several definitions.

Recall that $\mathbb{D}[0,T]$
is the set of all \cad functions
$f:[0,T]\rightarrow\mathbb{R}_{+}$.
Denote by $\tilde{\mathbb S}_t$ the canonical process
(i.e., $\tilde{\mathbb S}_t(\omega):=\omega_t$)
on $\mathbb{D}[0,T]$.
As usual, we consider the Borel $\sigma$--algebra
with respect to the sup norm {(this Borel $\sigma$--algebra coincides with the one generated
by the Skorohod topology)}.
Let $\tilde{\mathbb F}_t=\sigma\{\tilde{\mathbb S}_u|u\leq t\}$ be
the canonical filtration.

Let $\epsilon>0$, $n\in\mathbb N$ and
$\cT:=\{T_1,...,T_n,T\}$ be a
partition of the interval $[0,T]$, i.e.,
$0<T_1<...<T_n<T$.
{In the sequel we shall always
assume that $\epsilon<\ln (1+1/L)$
and $\epsilon<T_{i+1}-T_i$, $i=0,1,...,n-1$.}
\begin{dfn}
\label{d.measures}
{\rm{
For any $0<\tilde\kappa<\kappa$,
let
$\mathcal{M}^{\cT,\epsilon}_{\tilde\kappa,\cL}$ be the set of all
probability measures $\tilde{\mathbb Q}$ on the space
$\mathbb D[0,T]$ satisfying,}}
\begin{enumerate}
\item
{\rm{The canonical process $\tilde{\mathbb S}$ is of the form
$$\tilde{\mathbb S}_t=\sum_{i=0}^{n-1}\tilde{\mathbb S}_{\tilde\tau^{(\epsilon)}_k}\chi_{[\tilde\tau^{(\epsilon)}_k,\tilde\tau^{(\epsilon)}_{k+1})}+
\tilde{\mathbb S}_{\tilde\tau^{(\epsilon)}_n}\chi_{[\tilde\tau^{(\epsilon)}_k,\tilde\tau^{(\epsilon)}_{n+1}]},$$
where $0=\tilde\tau^{(\epsilon)}_0\leq \tilde\tau^{(\epsilon)}_1\leq...\leq\tilde\tau^{(\epsilon)}_{n+1}=T$ and $\tilde {\mathbb S}_0=1$.
}}
\item {\rm{
For any $k\leq n$, on the event $\tilde\tau^{(\epsilon)}_{k+1}<T$ we have
$$|\ln\tilde{\mathbb S}_{\tilde\tau^{(\epsilon)}_{k+1}}-\ln\tilde{\mathbb S}_{\tilde\tau^{(\epsilon)}_{k}}|=\epsilon.$$
}}
\item {\rm{For any $1\leq k\leq n+1$, $\tilde\tau^{(\epsilon)}_k\in \cT$,
$\tilde{\mathbb Q}$-a.s.}}
\item
{\rm{There exists a  $(\tilde{\mathbb Q}, \tilde{\mathbb F})$
$c\grave{a}dl\grave{a}g$ martingale
${\{\tilde M_t\}}_{t=0}^T$
such that}}
$$
(1-\tilde\kappa)\tilde{\mathbb S}_t\leq\tilde M_t\leq (1+\tilde\kappa)\tilde{\mathbb S}_t \ \ \tilde{\mathbb Q}\mbox{-a.s.};
$$
\item
{\rm{Finally,
\begin{align*}
\E_{\tilde{\Q}}[f_i(\tilde{\mathbb S})]&
\leq
\cL_i-L \hat C(e^{4\epsilon}+\epsilon-1), \ \ i=1,...,N-1, \\
E_{\tilde{\Q}}[f_N(\tilde{\mathbb S})] & \leq
\frac{\cL_N(1-L(e^{\epsilon}-1))-L\hat C(e^{\epsilon}-1)}{1+L(e^\epsilon-1)},
\end{align*}
where $\hat C:= 8 \sqrt{c_q^2+c_q \cL_N}$, and $c_q$ is given in \reff{2.100}.
}}
\end{enumerate}
\qed
\end{dfn}
The following result provides a lower bound on the super--replication price
$V^{\mathbb P}_{\kappa}(G)$.
\begin{lem}\label{lem3.1}
Let $\mathbb P$ be a probability measure on
$\Omega$ which satisfies the conditional full support property. Assume that
\begin{equation}\label{condition}
\min\left(\frac{1+\kappa}{1+\tilde\kappa},\frac{1-\tilde\kappa}{1-\kappa}\right)
\geq e^{2\epsilon}.
\end{equation}
Then, for every partition $\cT=\{T_1,\ldots, T_n,T\}$,
$$
V^{\mathbb P}_{\kappa}(G)
\geq \sup_{\tilde{\mathbb Q}\in \mathcal{M}^{\cT,\epsilon}_{\tilde\kappa,{\mathcal L}}}
\mathbb E_{\tilde{\mathbb Q}}[G(\tilde{\mathbb S})]-L\hat C (e^{4\epsilon}+\epsilon-1).
$$
{We always use the standard convention that the supremum over the empty set is minus infinity}.
\end{lem}
\begin{proof}
Fix, $\epsilon>0$ $\tilde \kappa$, $\cT$ as above.
{If $\mathcal{M}^{\cT,\epsilon}_{\tilde\kappa,{\cL}}=\emptyset$
then the statement is trivial.
Thus without  loss of generality we assume that $\mathcal{M}^{\cT,\epsilon}_{\tilde\kappa,{\cL}}\neq\emptyset$}.
We fix an arbitrary measure $\tilde{\mathbb Q}\in\
\mathcal{M}^{\cT,\epsilon}_{\tilde\kappa,{\cL}}$
and we will show that
\begin{equation}
\label{e.goal}
V^{\mathbb P}_{\kappa}(G)\geq
\mathbb E_{\tilde{\mathbb Q}}[G(\tilde{\mathbb S})]
-L\hat C (e^{4\epsilon}+\epsilon-1).
\end{equation}
The proof of the above inequality is completed in two steps. {In the first step we use the conditional full support
property of $\mathbb P$ and construct a consistent price system which is "close" to
$\tilde{\mathbb Q}$. In the second step we use the super--replication property and the constructed consistent
price system in order to obtain
a lower bound on the price.}
\vspace{8pt}

\noindent
\textbf{First step:} In this step,
we use the conditional full support property of
$\mathbb P$ in a similar way
to Guasoni, Rasonyi and Schachermayer \cite{GRS}.

Set
$
\tau^{(\epsilon)}_0:=\tau^{(\epsilon)}_0({\mathbb S})=0,
$
and for any positive integer $k>0$, recursively define
$$
\tau^{(\epsilon)}_k:=\tau^{(\epsilon)}_k({\mathbb S})=
T\wedge\inf \ \left\{t>\tau^{(\epsilon)}_{k-1}:
|\ln{\mathbb S}_t-\ln{\mathbb S}_{\tau^{(\epsilon)}_{k-1}}|=\epsilon\right\}
$$
where
as before we denote by $\Sb$ the canonical process on $\Omega$.
Define a random integer by,
$$
\mathbb K:=\mathbb K(\Sb)=\min\{k:\tau^{(\epsilon)}_k=T\}-1.
$$
Then, it is clear that $0\le \mathbb K <\infty$.
We also set,
$$
 S_k:=\mathbb{S}_{\tau^{(\epsilon)}_{k\wedge \mathbb K}}, \ \ 1\leq k\leq n+1,
 $$
 and
\begin{equation}
\label{e.sigmak}
 \sigma_k=\min\{t\in \cT: t\geq\tau^{(\epsilon)}_k\} .
 \end{equation}
Recall that the positive integer $n$ is the number of points in the
fixed partition $\cT=\{T_1,\ldots,T_n,T\}$.

For $\delta>0$, $i=1,...,n$ and $j=\pm1$,
let $g^{i,j}:[0,T_i]\rightarrow\mathbb{R}_{+}$,
be the linear functions satisfying
$$
g^{i,j}_0=1,\quad {\mbox{and}}
\quad
g^{i,j}_{T_i}=e^{\epsilon j}+ 2\delta j .
$$
We assume that $\delta$ is sufficiently small so that
$g^{i,j}$ is strictly positive.
Next, on $\Omega$ we define the events
\begin{eqnarray*}
A^{(j)}_i&:=&\{\
\sup_{0\leq t\leq T_i}\ |\mathbb S_t-g^{i,j}_t|<\delta\}, \qquad
i=1,...,n, \ j=\pm1\\
A^{(0)}_T&:=&
\{\sup_{0\leq t\leq T}|\ \mathbb S_t-1|<\delta\}.
\end{eqnarray*}
In view of the conditional full support property,
all of these events have non-zero
$\P$ probability.
Also, observe that for sufficiently small $\delta$,
for $i=1,\ldots,n$, $j=\pm 1$
$$
A^{(j)}_i
\subseteq B^{(j)}_i:=
\{\tau^{(\epsilon)}_1
\in  [T_i-\epsilon/n,T_i],\ \mathbb S_{\tau^{(\epsilon)}_1}
=\exp(\pm \epsilon)\}.
$$
Also $A^{(0)}_T\subset B^{(0)}_T:=\{\tau^{(\epsilon)}_1=T\}$.
Thus,  we conclude that
the events
$B^{(0)}_T, B^{(j)}_i$, $i=1,...,n$, $j=\pm 1$
have non-zero $\P$
probabilities as well.

We proceed by induction.
Assume that for a given $k\geq 1$ and any
$j_1,...,j_k=\pm 1$, $1\leq i_1<...<i_k\leq n$,
we have proved that the probability of the sets

$$
B^{(j_1,...,j_k)}_{i_1,...,i_k}
:=\bigcap_{m=1}^k
\left\{\tau^{(\epsilon)}_m\in  [T_{i_m}-\epsilon/n,T_{i_m}],
\ \mathbb S_{\tau^{(\epsilon)}_m}=\exp(\epsilon \sum_{r=1}^m j_r)\right\}
$$
and
 $$
 B^{(j_1,...,j_{k-1},0)}_{i_1,...,i_{k-1},T}:=
 \bigcap_{m=1}^{k-1}\left\{\tau^{(\epsilon)}_m\in  [T_{i_m}-\epsilon/n,T_{i_m}],
 \ \mathbb S_{\tau^{(\epsilon)}_m}
 =\exp(\epsilon \sum_{r=1}^m j_r)\right\}\bigcap
 \{\tau^{(\epsilon)}_k=T\}
 $$
have non-zero $\P$ probabilities.

Let
 $j_1,...,j_{k+1}=\pm 1$, $1\leq i_1<...<i_{k+1}\leq n$.
 On the event
 $\tau^{(\epsilon)}_k\leq T_{i_k}$ define the random,
 linear function
 $g^{i_{k+1},j_{k+1}}:[\tau^{(\epsilon)}_k,T_{i_{k+1}}]\rightarrow\mathbb{R}_{+}$
 by
$$
 g^{i_{k+1},j_{k+1}}_{\tau^{(\epsilon)}_k}=\exp(\epsilon \sum_{r=1}^{k}  j_r) \ \ \mbox{and} \ \
 g^{i_{k+1},j_{k+1}}_{T_{i_{k+1}}}=\exp(\epsilon\sum_{r=1}^{k+1}  j_r)+2\delta j_{k+1}.$$
 From the conditional full support property and
 Lemma 2.9 in Guasoni, Rasonyi and Schachermayer (2008), it follows that
 for any event $B\in\mathbb F_{\tau^{(\epsilon)}_k}$ the conditional probabilities
 $$
 \P\left(\ \sup_{\tau^{(\epsilon)}_k \leq t\leq T_{i_{k+1}}}
\left| \right.\mathbb S_t-g^{i_{k+1},j_{k+1}}_t
\left| \right.<\delta\
\left| \right. \ B^{(j_1,...,j_k)}_{i_1,...,i_k}\cap B\right)>0,
 $$
and
$$
\P\left(\ \sup_{\tau^{(\epsilon)}_k \leq t\leq T}
\left| \right.\mathbb S_t-\exp(\epsilon \sum_{r=1}^k j_r)
\left| \right.<\delta\
\left| \right. \ B^{(j_1,...,j_k)}_{i_1,...,i_k}\cap B\right)>0,
$$
provided that $\mathbb P( B^{(j_1,...,j_k)}_{i_1,...,i_k}\cap B)>0$.
Thus, similarly to the case $k=1$, for sufficiently small $\delta$
we conclude that the $\mathbb P$ probabilities of the following events
$$
B^{(j_1,...,j_{k+1})}_{i_1,...,i_{k+1}}:=\bigcap_{m=1}^{k+1}\left\{\tau^{(\epsilon)}_m
\in  [T_{i_m}-\epsilon/n,T_{i_m}],\
\mathbb S_{\tau^{(\epsilon)}_m}=\exp(\epsilon \sum_{r=1}^m j_r)\right\}
$$
and
 $$
 B^{(j_1,...,j_{k},0)}_{i_1,...,i_{k},T}:=
 \bigcap_{m=1}^{k}\left\{\tau^{(\epsilon)}_m
 \in  [T_{i_m}-\epsilon/n,T_{i_m}],\ \mathbb S_{\tau^{(\epsilon)}_m}
 =\exp(\epsilon \sum_{r=1}^m j_r)\right\}\bigcap
 \{\tau^{(\epsilon)}_{k+1}=T\}
 $$
are positive. This holds true for any $k\leq n+1$.

Recall the measure $\tilde{\mathbb Q}\in
\mathcal{M}^{\cT,\epsilon}_{\tilde\kappa,\cL}$
that was fixed at the start of the proof and the
$\sigma_k$'s defined by
\reff{e.sigmak}.
In view of the above discussion, and by using similar arguments
as in Lemma 2.4 in Guasoni, Rasonyi and Schachermayer (2008),
it follows that
there exists another probability measure
$\hat{\mathbb Q}\ll \mathbb P$ such that
the distribution of
$(S_1,...,S_{n+1},\sigma_1,...,\sigma_{n+1})$
under $\hat{\mathbb Q}$ is equal to the distribution of
$(\tilde{\mathbb S}_{\tilde\tau^{(\epsilon)}_1},\ldots,
\tilde{\mathbb S}_{\tilde\tau^{(\epsilon)}_{n+1}},
\tilde\tau^{(\epsilon)}_1,...,\tilde\tau^{(\epsilon)}_{n+1})$
under $\tilde{\mathbb Q}$, and in
addition
for any $i\leq n$, we have
\begin{equation}\label{3.1}
\hat{\mathbb Q}(\mathbb S_{i+1},\sigma_{i+1}|\mathbb{F}_{\tau^{(\epsilon)}_i})=
\hat{\mathbb Q}(\mathbb S_{i+1},\sigma_{i+1}|
\mathbb S_1,...,\mathbb S_i,\sigma_1,...,\sigma_i), \  \hat{\mathbb Q} \ \mbox{a.s.}
\end{equation}

Also observe that from our construction it follows that for any $k$,
\begin{equation}\label{3.1+}
|\sigma_k-\tau^{(\epsilon)}_k|\leq\frac{\epsilon}{n}, \ \ \hat{\mathbb Q} \ \ \mbox{a.s.}
\end{equation}
and
\begin{equation}\label{3.1++}
 S_{k+1}e^{-2\epsilon}\leq \mathbb S_t\leq S_{k+1}e^{2\epsilon}, \ \ \forall t\in[\tau^{(\epsilon)}_k,\tau^{(\epsilon)}_{k+1}]  \ \ \hat{\mathbb Q} \ \ \mbox{a.s.}
\end{equation}
Now, we arrive to the second step of the proof.
\vspace{10pt}

\textbf{Second step:}
Since $\tilde{\mathbb Q}\in \mathcal{M}^{\cT,\epsilon}_{\tilde\kappa,{\mathcal L}}$,
the definition of this set implies that there exists an associated martingale
 ${\{\tilde M_t\}}_{t=0}^T$ {which satisfies}
$$
(1-\tilde\kappa)\tilde{\mathbb S}_t\leq \tilde M_t\leq(1+\tilde\kappa)\tilde{\mathbb S}_t, \ \ t\in [0,T] \ \ \tilde{\mathbb Q} \ \ \mbox{a.s.}
$$

Then, for any $k\leq n+1$ there
exists a measurable function
$$
\psi_k:\mathbb{R}^{k}\times \cT\rightarrow \mathbb R_{+}
$$
such that
$$
\tilde{M}_{\tilde\tau^{(\epsilon)}_k}=
\psi_k(\tilde{\mathbb S}_{\tilde\tau^{(\epsilon)}_1},
\ldots,\tilde{\mathbb S}_{\tilde\tau^{(\epsilon)}_k},
\tilde\tau^{(\epsilon)}_1,...,\tilde\tau^{(\epsilon)}_k).
$$
{Moreover,
\begin{equation}\label{3.refnew}
(1-\tilde\kappa)\tilde{\mathbb S}_{\tilde\tau^{(\epsilon)}_k}\leq \tilde M_{\tilde\tau^{(\epsilon)}_k}\leq(1+\tilde\kappa)\tilde{\mathbb S}_{\tilde\tau^{(\epsilon)}_k}, \ \ k\leq n+1 \ \ \tilde{\mathbb Q} \ \ \mbox{a.s.}
\end{equation}}
Then, on $\Omega$ we define the stochastic process $M$ simply by
$$
M_k=\psi_k(S_1,...,S_k,\sigma_1,...,\sigma_k).
$$
In view of \reff{3.1} anf \reff{3.refnew}, it follows that for any $k$,
\begin{equation}\label{3.2}
\mathbb{E}_{\hat{\mathbb Q}}(M_{k+1}\ |\ \mathbb{F}_{\tau^{(\epsilon)}_k})=M_k
\end{equation}
and
\begin{equation}\label{3.3}
(1-\tilde\kappa)S_k\leq M_k\leq (1+\tilde\kappa)S_k \ \
\hat{\mathbb Q} \ \  \mbox{a.s.}
\end{equation}
Now, let $\pi=(c,\gamma)$ be a
$\mathbb P$ almost-surely super-replicating portfolio.
By \reff{condition}, \reff{3.1++}--\reff{3.3} and by summation by parts,
 it follows that
\begin{eqnarray}
\label{3.4}
&&
\mathbb{E}_{\hat{\mathbb Q}}\left(\gamma_T \mathbb S_T-
\kappa|\gamma_T|
\mathbb S_T +(1-\kappa) \int_{[0,T]}\mathbb S_u d\gamma^{-}_u-
(1+\kappa)\int_{[0,T]}\mathbb S_u d\gamma^{+}_u\right)\\
&&\nonumber
\hspace{10pt}\leq
\mathbb{E}_{\hat{\mathbb Q}}\left(\gamma_T M_{n+1}+
(1-\tilde\kappa) \sum_{k=0}^n S_{k+1}  \int_{[\tau^{(\epsilon)}_k,\tau^{(\epsilon)}_{k+1}]} d\gamma^{-}_u\right)\\
\nonumber
&&
\hspace{20pt}-
\mathbb{E}_{\hat{\mathbb Q}}
\left((1+\tilde\kappa)\sum_{k=0}^n S_{k+1}  \int_{[\tau^{(\epsilon)}_k,\tau^{(\epsilon)}_{k+1}]}
d\gamma^{+}_u\right)\\
\nonumber&&
\hspace{10pt}\leq
\mathbb{E}_{\hat{\mathbb Q}}\left(\gamma_T M_{n+1}+\sum_{k=0}^n M_{k+1}\left ( \int_{[\tau^{(\epsilon)}_k,\tau^{(\epsilon)}_{k+1}]} d\gamma^{-}_u
- \int_{[\tau^{(\epsilon)}_k,\tau^{(\epsilon)}_{k+1}]} d\gamma^{+}_u\right)\right)\\
\nonumber
&&\hspace{10pt} =\mathbb E_{\hat{\mathbb Q}}(\sum_{k=1}^n
\gamma_{\tau^{(\epsilon)}_k}(M_{k+1}-M_k))=0.
\end{eqnarray}
Next, we introduce the stochastic process ${\{\tilde S_t\}}_{t=0}^T$ by,
$$
\tilde{S}_t:=\sum_{k=0}^{n-1} S_k\chi_{[\sigma_k,\sigma_{k+1})}(t)+
S_n\chi_{[\sigma_n,T]}(t),
$$
where we set $\sigma_0=0$.
From our construction, it follows that the distribution
(on the space $\mathbb D[0,T]$) of
${\{\tilde S_t\}}_{t=0}^T$ under $\hat{\mathbb Q}$ is equal to
the distribution of $\tilde{\mathbb S}$ under $\tilde{\mathbb Q}$.
Thus,
\begin{equation}\label{3.5}
\mathbb{E}_{\hat{\mathbb Q}}G(\tilde S)=\mathbb{E}_{\tilde{\mathbb Q}}G(\tilde{\mathbb S})
\ \ \mbox{and} \ \
\mathbb{E}_{\hat{\mathbb Q}}f_i(\tilde S)=\mathbb{E}_{\tilde{\mathbb Q}} f_i(\tilde{\mathbb S}), \ \ i\leq N.
\end{equation}
{We next use the Assumption \ref{asm2.1}
and the properties \reff{3.1+}--\reff{3.1++}.
The result is the following inequalities that hold
$\hat{\mathbb Q}$ a.s.,
\begin{eqnarray}
\label{3.7}
|G(\tilde S)-G(\mathbb S)| & \leq & L(e^{4\epsilon}+\epsilon-1)
\|\tilde S\|,\\
\nonumber
|f_i(\tilde S)-f_i(\mathbb S)|
 & \leq & L(e^{4\epsilon}+\epsilon-1)\|\tilde S\|, \qquad
 \mbox{for} \ \ i\leq N-1.
\end{eqnarray}}
{
From Assumption \ref{asm.regularity} it follows that (recall that $e^{\epsilon}<\frac{L+1}{L}$) for any positive real numbers $x,y$
$$
|\ln x-\ln y|\leq \epsilon\Rightarrow q(y)\leq\frac{q(x)(1+L(e^{\epsilon}-1))+L (e^{\epsilon}-1)x}{1-L(e^{\epsilon}-1)}.$$
We conclude that
\begin{equation}\label{3.neww}
f_N(\mathbb S)\leq \frac{f_N(\tilde S)(1+L(e^{\epsilon}-1))+L (e^{\epsilon}-1)||\tilde S||}{1-L(e^{\epsilon}-1)}\ \ \hat{\mathbb Q} \ \ \mbox{a.s}.
\end{equation}
}
From (\ref{2.100}), Assumption \ref{asm.regularity} and the Doob inequality,
 it follows that
\begin{eqnarray*}
\mathbb{E}_{\hat{\mathbb Q}}[\|\tilde S\|^2] &=&
\mathbb{E}_{\tilde{\mathbb Q}}[\|\tilde{\mathbb  S}\|^2]
\leq 4 \mathbb{E}_{\tilde{\mathbb Q}}[\|\tilde M\|^2]\leq
16 \mathbb{E}_{\tilde{\mathbb Q}}[\tilde M ^2_T]\\
&\leq & 64 \mathbb{E}_{\tilde{\mathbb Q}}[\tilde{\mathbb S}^2_T]
\leq 64 [c_q^2+c_q \mathcal L_N]=\hat C^2,
\end{eqnarray*}
where the constants $\hat{C}$ and
$c_q$ are as in Definition \ref{d.measures}.
Also, the Holder inequality yields that
\begin{equation}\label{3.8}
\mathbb{E}_{\hat{\mathbb Q}}[\| \tilde S \|]\leq\hat C.
\end{equation}
Finally \reff{3.5}--\reff{3.8} and the fact that
$\tilde{\mathbb Q}\in \mathcal{M}^{\cT,\epsilon}_{\tilde\kappa,{\mathcal L}}$
implies that
$\mathbb{E}_{\hat{\mathbb Q}}f_i({\mathbb S})\leq \mathcal L_i$,
for every $i\leq N$.
Therefore, using \reff{3.4}--\reff{3.8}  and the relation
$\hat{\mathbb Q}\ll \mathbb P$ , we arrive at
$$
{\mathcal L}(c)\geq \mathbb{E}_{\hat{\mathbb Q}}
[c \cdot{f}(\mathbb S)]\geq
\mathbb{E}_{\hat{\mathbb Q}}[G(\mathbb S)]\geq
\mathbb{E}_{\tilde{\mathbb Q}}G[(\tilde{\mathbb S})]-L\hat C (e^{4\epsilon}+\epsilon-1).
$$
Since the above inequality holds for every $\P$ almost-surely
super-replicating strategy $\pi=(c,\gamma)$, this proves the
inequality \reff{e.goal} and completes the proof of this lemma.
\end{proof}

\section{Estimates for the Upper Bound} \label{sec4} \setcounter{equation}{0}
In this section we establish estimates
that will be used in the proof of the upper bound,
under the assumptions of Theorem \ref{thm.main}.

{We fix $\epsilon\in(0,\ln(1+1/L))$
and start with two definitions.}

\begin{dfn}
\label{d.DN}
{\rm{A function $F \in \mathbb{D}[0,T]$
belongs to $\mathbb{D}^{(\epsilon)}$, if it satisfies the followings,}}
\begin{enumerate}
\item $F_0=1$.
\item $F$ {\rm{is piecewise constant with jumps
at times $t_1,...,t_n$, where}}
$$
t_0=0<t_1<t_2<...<t_n<T.
$$
\item {\rm{For  any}} $k=1,...,n$, $|\ln F_{t_k}-\ln F_{t_{k-1}}|=\epsilon$.
\item{\rm{ For  any $k=1,...,n$, $t_k-t_{k-1}\in  U^{(\epsilon)}_k$,
where
$$
U^{(\epsilon)}_k:=
\left\{i\epsilon/(2^{k}): i=1,2, \ldots, \right\} \cup
\left\{\epsilon/(i2^k ) : i=1,2, \ldots, \right\},
$$
are the sets of possible differences between
two consecutive jump times. We emphasise, in the fourth condition,
the dependence
of the set $U^{(\epsilon)}_k$ on $k$.  So
as $k$ gets larger, jump times
take values in a finer grid.}}
\end{enumerate}
\qed
\end{dfn}

\begin{dfn}\label{dfn4.1}
{\rm{
For $\tilde\kappa,\Lambda>0$, let
$\mathcal{M}^{\epsilon,\Lambda}_{\tilde\kappa,{\cL}}$
be the set of all probability measures
$\tilde{\mathbb Q}$ on the space
$\mathbb D[0,T]$ such that the following holds,}}
\begin{enumerate}
\item {\rm{The probability measure $\tilde{\mathbb Q}$ is supported on the set}}
$\mathbb{D}^{(\epsilon)}$.
\item{\rm{
There exists a $c\grave{a}dl\grave{a}g$
$(\tilde{\mathbb Q},  \tilde{\mathbb F})$ martingale
${\{\tilde M_t\}}_{t=0}^T$  such that}}
$$
(1-\tilde\kappa)\tilde{\mathbb S}_t
\leq\tilde M_t
\leq (1+\tilde\kappa)\tilde{\mathbb S}_t
 \  \ \tilde{\mathbb Q} \ \ \mbox{a.s.}
$$
\item{\rm{
Let $\hat{C}$ be as in Definition \ref{d.measures}
and $L$ be as in Assumption \ref{asm2.1}.
Set}}
$$
B:=L(e^{2\epsilon}+\epsilon-1)
\frac{\hat C^2}{2(1-8\kappa)}+2L(e^{\epsilon}-1)\mathcal L_N+\epsilon
$$
{\rm{For any $i<N$,}}
$$
\mathbb E_{\tilde{\mathbb Q}}[f_i(\tilde{\mathbb S})]\leq {\mathcal L}_i+B,
$$
{\rm{and}}
$$
\mathbb E_{\tilde{\mathbb Q}}[f_N(\tilde{\mathbb S})
\wedge \Lambda(\tilde{\mathbb S}_T+1)]\leq {\mathcal L}_N+B.
$$
\end{enumerate}
\qed
\end{dfn}

The following result provides an upper bound on the model--free super--replication price
$V_{\kappa}(G)$.
\begin{lem}\label{lem4.1}
Assume that
\begin{equation}\label{condition1}
\min\left(\frac{1+\tilde\kappa}{1+\kappa},\frac{1-\kappa}{1-\tilde\kappa}\right)\geq e^{4\epsilon}.
\end{equation}
 Then
 $$
 V_{\kappa}(G)\leq \left(\sup_{\tilde{\mathbb Q}\in
 \mathcal{M}^{\epsilon,\Lambda}_{\tilde\kappa,{\cL}}}
\mathbb E_{\tilde{\mathbb Q}}[G(\tilde{\mathbb S})]\right)^{+}+L(e^{2\epsilon}+\epsilon-1)
\frac{\hat C^2}{2(1-8\kappa)}.
$$
\end{lem}
Again, we use the standard convention that
the supremum over the empty set is
minus infinity.  In particular,
if $ \mathcal{M}^{\epsilon,\Lambda}_{\tilde\kappa,{\cL}}$ is empty,
then the above lemma
states that  $ V_{\kappa}(G)\leq  L(e^{2\epsilon}+\epsilon-1)
\frac{\hat C^2}{2(1-8\kappa)}$.

\begin{proof}
The proof is completed in two steps.
In the first step, we apply the results that deal with the
``classical" super--replication with proportional transaction costs.
\vspace{5pt}

\noindent
\textbf{First step:}
Since $\mathbb D^{(\epsilon)}$ is countable,
there exists a
probability measure $\tilde{\mathbb P}$
satisfying
$\tilde{\mathbb P}(\mathbb D^{(\epsilon)})=1$
and $\tilde{\mathbb P}(\{F\})>0$ for all $F\in \mathbb D^{(\epsilon)}$. Consider the filtered probability space
$(\mathbb D[0,T], {\{\tilde{\mathbb F}_t\}}_{t=0}^T,\tilde{\mathbb F}_T,\tilde{\mathbb P})$.
Denote by $\mathcal M_{\tilde\kappa}$ the set of all
consistent price systems in $\mathbb D^{(\epsilon)}$. Namely,
$\tilde{\mathbb Q}\in \mathcal M_{\tilde\kappa}$ if
$\tilde{\mathbb Q}$ is equivalent to $\tilde{\mathbb P}$
and
there exists a $c\grave{a}dl\grave{a}g$ martingale ${\{\tilde M_t\}}_{t=0}^T$ (with respect to $\tilde{\mathbb Q}$ and $\tilde{\mathbb F})$
such that
$$(1-\tilde\kappa)\tilde{\mathbb S}_t\leq\tilde M_t\leq (1+\tilde\kappa)\tilde{\mathbb S}_t \ \ \tilde{\mathbb P} \ \  \mbox{a.s.}$$
Let $X:=X(\tilde{\mathbb S})$ be random variable which is $\tilde{\mathbb F}_T$ measurable and bounded from below by a multiple of
$1+\tilde{\mathbb S}_T$.
Set
\begin{equation}\label{4.1}
c_0:=\sup_{\tilde{\mathbb Q}\in \mathcal{M}_{\tilde\kappa}}
\mathbb E_{\tilde{\mathbb Q}} [X].
\end{equation}
From Theorem 1.5 in Schachermayer \cite{S},
 it follows that there exists a predictable stochastic process of bounded variation
${\{\tilde\gamma_t\}}_{t=0}^T$ such that $\tilde\gamma_0=\tilde\gamma_T=0$ and
%\begin{equation}\label{4.1+}
$$
c_0+(1-\tilde\kappa) \int_{[0,T]}{\tilde {\mathbb S}}_ud{\tilde\gamma}^{-}_u-
(1+\tilde\kappa)\int_{[0,T]}{\tilde{\mathbb S}}_ud{\tilde\gamma}^{+}_u\geq X, \ \ \tilde{\mathbb P} \ \ \mbox{a.s.}
$$
%\end{equation}
Thus, there exists a predictable
map $\tilde\gamma:\mathbb{D}^{(\epsilon)}\rightarrow\mathbb L^{\infty}[0,T]$ such that for any $F\in \mathbb D^{(\epsilon)}$
$\tilde\gamma_0(F)=\tilde\gamma_T(F)=0$ and
\begin{equation}\label{4.2}
c_0+(1-\tilde\kappa) \int_{[0,T]}F_ud{\tilde\gamma}^{-}_u(F)-
(1+\tilde\kappa)\int_{[0,T]}F_ud{\tilde\gamma}^{+}_u(F)\geq X(F),
\end{equation}
{where $\mathbb L^{\infty}[0,T]$ is the set of all bounded functions on the interval $[0,T]$}.
Next, choose $(c_1,...,c_N)\in\mathbb{R}^N_{+}$
and consider the random variable
$$X=X(\tilde{\mathbb S})=G(\tilde{\mathbb S})-\sum_{i=1}^{N-1} c_i f_i(\tilde{\mathbb S})-c_N (f_N(\tilde{\mathbb S})\wedge \Lambda (\tilde{\mathbb S}_T+1)).$$
{Recall, that in Assumption \ref{asm.regularity} we assumed
that if $f_i$ is path dependent then it is bounded. This together with the Lipschitz continuity of $f_i$, $i=1,...,N-1$ yields that $f_1(\tilde{\mathbb S}),...,f_{N-1}(\tilde{\mathbb S})$ are bounded by a multiple
of $1+\tilde{\mathbb S}_T$, and so $X$ is bounded by
a multiple
of $1+\tilde{\mathbb S}_T$ as well.}

Let $(c_0,\tilde\gamma)$ be such that \reff{4.1} and \reff{4.2} hold true.

Next, we lift the trading strategy $\tilde\gamma$
to a trading strategy on the space $\Omega$. We start with some preparations.
Recall the definition of the stopping times
$\tau^{(\epsilon)}_k:=\tau^{(\epsilon)}_k(\mathbb S)$,
$k\geq 0$, and $\mathbb K:=\mathbb K(\mathbb S)=
\min\{k:\tau^{(\epsilon)}_k=T\}-1$.

Set,
\begin{eqnarray*}
\hat \tau^{(\epsilon)}_k&:=&\sum_{i=1}^k\ \Delta \hat \tau^{(\epsilon)}_i, \ \ \mbox{where}\\
\Delta \hat \tau^{(\epsilon)}_i&=& \max\{\Delta t \in
U^{(\epsilon)}_i: \Delta t< \Delta \tau^{(\epsilon)}_i:= \tau^{(\epsilon)}_i-\tau^{(\epsilon)}_{i-1}\}.
\end{eqnarray*}
It is clear that
$0=\hat \tau^{(\epsilon)}_0<\hat \tau^{(\epsilon)}_1<...<\hat \tau^{(\epsilon)}_{\mathbb K}<T$
and $\hat \tau^{(\epsilon)}_k < \tau^{(\epsilon)}_k$ for all  $k=0,\ldots,\mathbb K$.

We now define $\Psi:\Omega\rightarrow\mathbb{D}^{(\epsilon)}$ by
$$
\Psi_t(\mathbb S):=\sum_{k=0}^{\mathbb K-1}
\ \mathbb{S}_{\tau^{(\epsilon)}_{k}} {\chi}_{[\hat \tau^{(\epsilon)}_k,\hat \tau^{(\epsilon)}_{k+1})}(t)+
\mathbb{S}_{\tau^{(\epsilon)}_{\mathbb K}} {\chi}_{[\hat \tau^{(\epsilon)}_{\mathbb K},T]}(t).$$
Finally, define the hedge $\pi=(c,\gamma)$ where $c=(c_0,c_1,...,c_N)$ and
$$
\gamma(\mathbb S):=
\sum_{k=1}^{\mathbb K}\tilde\gamma_{\hat\tau^{(\epsilon)}_k}(\Psi(\mathbb S))
{\chi}_{ (\tau^{(\epsilon)}_k,\tau^{(\epsilon)}_{k+1}]}(t).
$$

We continue
by estimating the portfolio value $Z^\pi_T(\mathbb S)$.
Set
\begin{eqnarray*}
I&:=&I(\mathbb S)=\gamma_T\mathbb S_T-\kappa|\gamma_T|\mathbb S_T+(1-\kappa)\int_{[0,T]}\mathbb S_u d\gamma^{-}_u-(1+\kappa)\int_{[0,T]}\mathbb S_u d\gamma^{+}_u\\
&&-(1-\tilde\kappa)\int_{[0,T]}\Psi_u(\mathbb S) d\tilde\gamma^{-}_u(\Psi(\mathbb S))+
(1+\tilde\kappa)\int_{[0,T]}\Psi_u(\mathbb S) d\tilde\gamma^{+}_u(\Psi(\mathbb S)).
\end{eqnarray*}
From Assumption \ref{asm.regularity}
it follows that for any $x,y>0$
$$|\ln x-\ln y|<\epsilon \Rightarrow q(x)
\geq \frac {(1-L(e^{\epsilon}-1))q(y)-L(e^{\epsilon}-1)y}{1+L(e^{\epsilon}-1)}.$$
Thus, from Assumptions \ref{asm2.1}, \ref{asm.regularity}
and \reff{4.2}, it follows that
\begin{eqnarray}
\label{4.3}
&&
Z^\pi_T(\mathbb S)-G(\mathbb S)\geq I-(G(\mathbb S)
-G(\Psi(\mathbb S)))
-\sum_{i=1}^{N} c_i(f_i(\Psi(\mathbb S))-f_i(\mathbb S))\\
&&\geq
I-L \left(1+\sum_{i=1}^{N-1} c_i\right)
\left(e^{2\epsilon}+\sum_{j=1}^\infty \epsilon 2^{-j}-1\right)
\|\mathbb S\|-L c_N(e^{\epsilon}-1) \frac{2 f_N(\mathbb S)+\|S\|}{1+L(e^{\epsilon}-1)}
\nonumber\\
\vspace*{2in}
&&\geq
 I-L \left(1+\sum_{i=1}^{N-1} c_i\right)
 \left(e^{2\epsilon}+\epsilon-1)\|\mathbb S\|- L c_N(e^{\epsilon}-1\right)
 (2 f_N(\mathbb S)+\|S\|).\nonumber
\end{eqnarray}
It remains to estimate the term $I$.
To simplify the calculations, we use the notation $\gamma=\gamma(\mathbb S)$
and  $\tilde\gamma=\tilde\gamma(\Psi(\mathbb S))$.
Then, in view of \reff{condition1},
\begin{eqnarray*}
&&\gamma_T\mathbb S_T-\kappa|\gamma_T|\mathbb S_T
+(1-\kappa)\int_{[0,T]}\mathbb S_u d\gamma^{-}_u-(1+\kappa)\int_{[0,T]}
\mathbb S_u d\gamma^{+}_u
\\
&&\hspace{8pt}
\geq \gamma_T\mathbb S_T-\kappa|\gamma_T|\mathbb S_T
+\sum_{k=1}^{\mathbb K}\mathbb S_{\tau^{(\epsilon)}_{k-1}}
\int_{[\tau^{(\epsilon)}_k,\tau^{(\epsilon)}_{k+1}]}
[(1-\tilde\kappa)d\gamma^{-}_u-(1+\tilde\kappa)d\gamma^{+}_u]
\\
&&\hspace{8pt}
= \gamma_T\mathbb S_T-\kappa|\gamma_T|\mathbb S_T
+\sum_{k=1}^{\mathbb K}\mathbb S_{\tau^{(\epsilon)}_{k-1}}
\int_{[\tau^{(\epsilon)}_k,\tau^{(\epsilon)}_{k+1}]}
[-d\gamma_u-\tilde\kappa|d\gamma_u|]\\
&&\hspace{8pt}
\geq\gamma_T\mathbb S_T-\kappa|\gamma_T|\mathbb S_T
+\sum_{k=0}^{\mathbb K-1}\Psi_{\hat\tau^{(\epsilon)}_k}(\mathbb S)
\int_{[\hat\tau^{(\epsilon)}_k,\hat\tau^{(\epsilon)}_{k+1}]}
[-d\tilde\gamma_u-\tilde\kappa|d\tilde\gamma_u|]
\\
&&\hspace{8pt}
=\gamma_T\mathbb S_T-\kappa|\gamma_T|\mathbb S_T
+(1-\tilde\kappa)\int_{[0,\hat\tau^{(\epsilon)}_{\mathbb K}]}\Psi_{u}(\mathbb S) d\tilde\gamma^{-}_u-
(1+\tilde\kappa)\int_{[0,\hat\tau^{(\epsilon)}_{\mathbb K}]}\Psi_{u}(\mathbb S) d\tilde\gamma^{+}_u
\\
&&\hspace{8pt}
\geq (1-\tilde\kappa)\int_{[0,T]}\Psi_u(\mathbb S) d\tilde\gamma^{-}_u-
(1+\tilde\kappa)\int_{[0,T]}\Psi_u(\mathbb S) d\tilde\gamma^{+}_u.
\end{eqnarray*}
Hence,
we conclude that $I\geq 0$.
We use this inequality together with \reff{2.200} and \reff{4.3}.
The result is,
\begin{eqnarray*}
V_{\kappa}(G)&\leq& {\mathcal L}(c)+ L(e^{2\epsilon}+\epsilon-1)
(1+ \sum_{i=1}^N c_i) V_{\kappa}(\|\mathbb S\|)+
2 L(e^{\epsilon}-1)c_N V_{\kappa}(f_N(\mathbb S)) \\
&\leq & {\mathcal L}(c)+ L(e^{2\epsilon}+\epsilon-1)
\frac{\hat C^2}{2(1-8\kappa)}(1 +\sum_{i=1}^N c_i) +
2 L(e^{\epsilon}-1)c_N \mathcal L_N.
\end{eqnarray*}
This together with \reff{4.1} yields
\begin{equation}
\label{4.5}
V_{\kappa}(G)\leq\inf_{c_1,...,c_N\geq 0}\sup_{\tilde{\mathbb Q}\in \mathcal{M}_{\tilde\kappa}}\left(\mathbb E_{\tilde{\mathbb Q}}[\ \xi \ ]+\sum_{i=1}^N c_i A_i\right)
+L(e^{2\epsilon}+\epsilon-1)
\frac{\hat C^2}{2(1-8\kappa)},
\end{equation}
where
\begin{eqnarray*}
\xi&:=& G(\tilde{\mathbb S})-
\sum_{i=1}^{N-1} c_i f_i(\tilde{\mathbb S})-c_N (f_N(\tilde{\mathbb S})\wedge \Lambda(\tilde{\mathbb S}_T+1)),\\
A_i&:=&\mathcal L_i+ L(e^{2\epsilon}+\epsilon-1)
\frac{\hat C^2}{2(1-8\kappa)}+2L(e^{\epsilon}-1)\mathcal L_N
=\mathcal L_i+B-\epsilon, \ \ i\leq N.
\end{eqnarray*}

\textbf{Second Step:}
The next step is to interchange
 the order of the infimum and supremum in \reff{4.5}.
 Consider the compact set
 $H:=[0,K/\epsilon]^N$, where recall $K$ is satisfying $G\leq K$.
 Define the function
 $\mathcal G:H\times \mathcal{M}_{\tilde\kappa}\rightarrow\mathbb R$
 by
 $$
 \mathcal G(h,\tilde{\mathbb Q})=
 \mathbb E_{\tilde{\mathbb Q}}\left[
 G(\tilde{\mathbb S})-\sum_{i=1}^{N-1}
 h_i f_i(\tilde{\mathbb S})
 -h_N (f_N(\tilde{\mathbb S})\wedge \Lambda(\tilde{\mathbb S}_T+1))
 \right]
 +\sum_{i=1}^N h_i A_i,$$
 where $h=(h_1,...,h_N)$.
Notice that $\mathcal G$ is affine in each of the variables, and continuous in the
first variable.
The set $ \mathcal{M}_{\tilde\kappa}$ can be naturally considered as a subset of the vector space
$\mathbb{R}^{\mathbb{D}^{(\epsilon)}}$. Let us show that
$ \mathcal{M}_{\tilde\kappa}$ is a convex set. Let
$\tilde{\mathbb Q}_1,\tilde{\mathbb Q}_2\in \mathcal{M}_{\tilde\kappa}$ and let $\lambda\in (0,1)$.
Consider the measure $\tilde{\mathbb Q}=\lambda \tilde{\mathbb Q}_1+(1-\lambda)\tilde{\mathbb Q}_2$.
For $i=1,2$ let ${\{\tilde M^{(i)}_t\}}_{t=0}^T$ be a martingale with respect
to $\tilde{\mathbb Q}_i$ and $\tilde{\mathbb F}$,
such that
$$(1-\tilde\kappa)\tilde{\mathbb S}_t\leq\tilde M^{(i)}_t\leq (1+\tilde\kappa)\tilde{\mathbb S}_t \ \ \tilde{\mathbb P} \ \  \mbox{a.s.}$$
Define the stochastic process
$$
\tilde M_t=\lambda \tilde M^{(1)}_t\
\left[
\frac{d \tilde{\mathbb Q}_1}{d \tilde{\mathbb Q}}|{\tilde{\mathbb F}_t}\right]
+ (1-\lambda)\tilde  M^{(2)}_t\
\left[ \frac{d \tilde{\mathbb Q}_2}{d \tilde{\mathbb Q}}|{\tilde{\mathbb F}_t}\right],
\ \ t\in [0,T].
$$
Clearly, ${\{\tilde M_t\}}_{t=0}^T$ is a martingale with respect
to $\tilde{\mathbb Q}$ and $\tilde{\mathbb F}$. Also, since
$\tilde M_t$ is a (random) convex combination of
$\tilde M^{(1)}_t$ and $\tilde M^{(2)}_t$,
$$
(1-\tilde\kappa)\tilde{\mathbb S}_t\leq\tilde M_t\leq (1+\tilde\kappa)\tilde{\mathbb S}_t \ \ \tilde{\mathbb P} \ \  \mbox{a.s.}
$$
Hence, $\tilde{\mathbb Q}\in \mathcal{M}_{\tilde\kappa}$,.
This proves that $\mathcal{M}_{\tilde\kappa}$ is a convex set.
Next, we apply the min--max theorem,  Theorem 2,
in  Beiglb\"ock, Henry-Labord\`ere and Penkner \cite{BHLP}
to $\mathcal G$.  The result is,
$$
\inf_{h\in H}\sup_{\tilde{\mathbb Q}\in \mathcal{M}_{\tilde\kappa}}
\mathcal G(h,\tilde{\mathbb Q})=
\sup_{\tilde{\mathbb Q}\in \mathcal{M}_{\tilde\kappa}}\inf_{h\in H} \mathcal G(h,\tilde{\mathbb Q})\leq
\sup_{\tilde{\mathbb Q}\in \mathcal{M}_{\tilde\kappa}} \mathcal G(h^{\tilde{\mathbb Q}},\tilde{\mathbb Q}),
$$
where
$$
h^{\tilde{\mathbb Q}}_i=\frac{K}{\epsilon}{\chi}_{
\left\{\mathbb{E}_{\tilde{\mathbb Q}}[f_i({\tilde{\mathbb S}})]\geq \mathcal L_i+B\right\}},
\ \ i\leq N-1, \ \
h^{\tilde{\mathbb Q}}_N=\frac{K}{\epsilon}{\chi}_{\left\{\mathbb{E}_{
\tilde{\mathbb Q}}
[f_N({\tilde{\mathbb S}})\wedge\Lambda(\tilde{\mathbb S}_T+1)]\geq \mathcal L_N+B
\right\}}.
$$
The definitions of $h^{\tilde{\mathbb Q}}$, the set
$\mathcal{M}^{\epsilon,\Lambda}_{\tilde\kappa,{\cL}}$ and the fact that $G\leq K$
implies that
$$
\mathcal G(h^{\tilde{\mathbb Q}},\tilde{\mathbb Q}) \le 0,
\qquad \forall \ \tilde{\mathbb Q}\in \mathcal{M}_{\tilde\kappa}\
{\mbox{but}}\ \tilde{\mathbb Q} \not \in \mathcal{M}^{\epsilon,\Lambda}_{\tilde\kappa,{\cL}}.
$$
In particular, $\sup_{\tilde{\mathbb Q}\in \mathcal{M}_{\tilde\kappa}}
\mathcal G(h^{\tilde{\mathbb Q}},\tilde{\mathbb Q})
\leq 0$, if the set
$\mathcal{M}^{\epsilon,\Lambda}_{\tilde\kappa,{\cL}}$ is empty.
These together with \reff{4.5} implies that
\begin{align*}
V_{\kappa}(G)&\leq
\sup_{\tilde{\mathbb Q}\in \mathcal{M}_{\tilde\kappa}}
\mathcal G(h^{\tilde{\mathbb Q}},\tilde{\mathbb Q})+L(e^{2\epsilon}+\epsilon-1)
\frac{\hat C^2}{2(1-8\kappa)}\\
&\leq
\left(\sup_{\tilde{\mathbb Q}\in \mathcal{M}^{\epsilon,\Lambda}_{\tilde\kappa,{\cL}}}
\mathbb E_{\tilde{\mathbb Q}}[G(\tilde{\mathbb S})]\right)^{+}+L(e^{2\epsilon}+\epsilon-1)
\frac{\hat C^2}{2(1-8\kappa)}.
\end{align*}
\end{proof}

\section{Asymptotical Analysis of the Bounds}
\label{sec5} \setcounter{equation}{0}
In this section. we complete the proof of Theorem \ref{thm.main}.
This is achieved by proving that
the lower and the upper bounds from Sections \ref{sec3}
and \ref{sec4} are asymptotically equal to each other.

Recall the probability measure
$\mathbb Q$ from Assumption \ref{asm.noarbitrage}.
Set,
$D_i=\mathbb E_{\mathbb Q}[f_i(\mathbb S)]$, $i\leq N$. Denote
$\mathbf D=\prod_{i=1}^N (D_i,\infty)$.
Let $H=(H_1,...,H_N)\in\mathbf D$ and let $\tilde\kappa\in (0,1)$. Define
$\mathcal{M}_{\tilde\kappa,H}$ to be the set of all probability measures on
$\hat\Omega:=\Omega\times \mathcal{C}^{++}_{[0,T]}$
which satisfy the conditions of Definition \ref{dfn.1},
with $\kappa,\mathcal L_1,...,\mathcal L_N$
replaced by $\tilde\kappa, H_1,...,H_N$. {Observe that $\mathbb Q\in\mathcal{M}_{\tilde\kappa,H}$ and so, the
set $\mathcal{M}_{\tilde\kappa,H}$ is not empty.}
Define the function
$\Gamma:\mathbf D\times (0,1)\rightarrow\mathbb R$ by
 $$
 \Gamma(H,\tilde\kappa):=\sup_{\hat{\mathbb Q}\in
 \mathcal{M}_{\tilde\kappa,H}}\mathbb E_{\hat{\mathbb Q}}[G(\mathbb S^{(1)})],
 $$
where, recall the canonical process
 $\hat{\mathbb S}=(\mathbb S^{(1)}_t,\mathbb{S}^{(2)}_t)_{0\leq t\leq T}$ given
 in Definition \ref{dfn.1}. The following lemma is central
  in the analysis of the asymptotic behaviour of the bounds.

 \begin{lem}\label{lem5.1}
 The function $\Gamma:\mathbf D\times (0,1)\rightarrow\mathbb R$ is continuous.
 \end{lem}
\begin{proof}
It suffices to prove
 that for any compact set $J\subset\mathbf D\times (0,1)$ there exists a
a continuous function $m_J:\mathbb {R}_{+}\rightarrow\mathbb{R}_{+}$
(modulus of continuity)
with $m_J(0)=0$ such that for any $(H^{(i)},\tilde\kappa_i)\in J$, $i=1,2$
$$
\Gamma(H^{(1)},\tilde\kappa_1)-\Gamma(H^{(2)},\tilde\kappa_2)
\leq m_J\left(\sum_{k=1}^N |H^{(1)}_k-H^{(2)}_k|+|\tilde\kappa_1-\tilde\kappa_2|\right).
$$

Choose $\epsilon>0$.
There exists $\hat{\mathbb Q}_1\in \mathcal{M}_{\tilde\kappa_1,H^{(1)}}$
such that
\begin{equation}\label{5.1}
\Gamma(H^{(1)},\tilde\kappa_1)<\epsilon
+\mathbb E_{\hat{\mathbb Q}_1}[G(\mathbb S^{(1)})].
\end{equation}
On the space $\hat\Omega$, define the stochastic processes $\rho$
and $\dot{\rho}$ by,
$$
\rho_t:=\frac{\mathbb {S}^{(2)}_t}{\mathbb{S}^{(1)}_t} \ \ \mbox{and} \ \
\dot{\rho}_t:=(1-\tilde\kappa_2)\vee (\rho_t\wedge{(1+\tilde\kappa_2)}), \ \ t\in [0,T].
$$
Next, introduce the stochastic process
$\dot{\mathbb S}=(\dot{\mathbb S}^{(1)}_t,\dot{\mathbb{S}}^{(2)}_t)_{0\leq t\leq T}$ by
$$
\dot{\mathbb S}^{(1)}_t:=\frac{{\mathbb S}^{(2)}_t}{\dot{\rho_t}}\frac{\dot{\rho_0}}{\rho_0}
=\frac{\rho_t}{\dot{\rho}_t}\frac{\dot{\rho}_0}{\rho_0}
{\mathbb S}^{(1)}_t \ \ \mbox{and} \ \
\dot{\mathbb S}^{(2)}_t:=\frac{\dot{\rho_0}}{{\rho}_0}{\mathbb S}^{(2)}_t, \ \ t\in [0,T].
$$
{Observe that there exists a constant $C^{(1)}_J$ such that
\begin{equation}\label{5.new}
\sup_{0\leq t\leq T}|\ln\dot{\mathbb S}^{(1)}_t-\ln {\mathbb S}^{(1)}_t|=
\sup_{0\leq t\leq T}
|\ln\rho_t+\ln \dot\rho_0-\ln\dot\rho_t-\ln\rho_0|
\leq
C^{(1)}_J |\tilde\kappa_1-\tilde\kappa_2|.
\end{equation}
Without loss of generality we assume that
$C^{(1)}_J |\tilde\kappa_1-\tilde\kappa_2|<\ln(1+1/L)$.}

{The idea behind the definition of the process $\dot{\mathbb S}$ is to construct a stochastic process
which is "close" to $\mathbb S$
and satisfy properties (1) and (2) of Definition \ref{dfn.1}, for
$\tilde\kappa_2$ instead of $\tilde\kappa_1$. In addition we require that
$\dot{\mathbb S}^{(1)}_0=1$.
Indeed, observe that $\dot{\mathbb S}:\hat\Omega\rightarrow\hat\Omega$. Thus, define
the probability measure $\hat{\mathbb Q}_2$ to be the distribution of
$\dot{\mathbb S}$ under the probability measure $\hat{\mathbb Q}_1$.
Namely,
 $\hat{\mathbb Q}_2$ is a probability
measure
on $\hat\Omega$ which is given by
$\hat{\mathbb Q}_2(A)=\hat{\mathbb Q}_1(\dot{\mathbb S}^{-1}(A))$
for any Borel set $A\subset\hat\Omega$.}
 Clearly, for any $t\in [0,T]$
$$(1-\tilde\kappa_2)\dot{\mathbb S}^{(1)}_t\leq \dot{\mathbb S}^{(2)}_t\leq (1+\tilde\kappa_2)\dot{\mathbb S}^{(1)}_t, \ \ \hat{\mathbb Q}_1 \ \  \mbox{a.s.}$$
and
$$\mathbb E_{\hat{\mathbb Q}_1}({\dot{\mathbb S}}^{(2)}_T|\dot{\mathbb S}_u, \ u\leq t)={\dot{\mathbb S}}^{(2)}_t.$$
Thus, for any $t\in [0,T]$,
\begin{equation}\label{5.1+}
(1-\tilde\kappa_2){\mathbb S}^{(1)}_t\leq {\mathbb S}^{(2)}_t\leq (1+\tilde\kappa_2){\mathbb S}^{(1)}_t, \ \ \hat{\mathbb Q}_2 \ \  \mbox{a.s.}
\end{equation}
and
\begin{equation}\label{5.1++}
\mathbb E_{\hat{\mathbb Q}_2}({{\mathbb S}}^{(2)}_T|\hat{\mathbb F}_t)={{\mathbb S}}^{(2)}_t.
\end{equation}
Next, similarly to \reff{3.8} we obtain that
there exists a constant $C^{(2)}_J$ such that
$$\mathbb E_{\hat{\mathbb Q}_1}[||\mathbb S^{(1)}||]\leq C^{(2)}_J.$$

{By applying Assumptions \ref{asm2.1}--\ref{asm.regularity} in a similar way to
 \reff{3.7}--\reff{3.neww}, and
using (\ref{5.new})
we obtain that we can construct another constant $C^{(3)}_J$ satisfying,
\begin{align}
\nonumber
|\mathbb E_{\hat{\mathbb Q}_2}[G(\mathbb S^{(1)})]-
\mathbb E_{\hat{\mathbb Q}_1}[G(\mathbb S^{(1)})]| & =
|\mathbb E_{\hat{\mathbb Q}_1}[G(\dot{\mathbb S}^{(1)})]-
\mathbb E_{\hat{\mathbb Q}_1}[G(\mathbb S^{(1)})]|\\
\label{5.2}
&\leq L  C^{(2)}_J (\exp(C^{(1)}_J |\tilde\kappa_1-\tilde\kappa_2|)-1)\\
&\leq
C^{(3)}_J|\tilde\kappa_1-\tilde\kappa_2|\nonumber\\
\nonumber
|\mathbb E_{\hat{\mathbb Q}_2}[f_i(\mathbb S^{(1)})]-
\mathbb E_{\hat{\mathbb Q}_1}[f_i(\mathbb S^{(1)})]| &=
|\mathbb E_{\hat{\mathbb Q}_1}[f_i(\dot{\mathbb S}^{(1)})]-\mathbb E_{\hat{\mathbb Q}_1}[f_i(\mathbb S^{(1)})]|\\
\label{5.3}
&\leq L  C^{(2)}_J (\exp(C^{(1)}_J |\tilde\kappa_1-\tilde\kappa_2|)-1)\\
&\leq
C^{(3)}_J|\tilde\kappa_1-\tilde\kappa_2|, \ \ i\leq N-1,\nonumber
\end{align}
and for $i=N$
\begin{align}
\label{5.newww}
|\mathbb E_{\hat{\mathbb Q}_2}[f_N(\mathbb S^{(1)})]&=
\mathbb E_{\hat{\mathbb Q}_1}[f_N(\dot{\mathbb S}^{(1)})]|\\
& \leq\frac{\mathbb E_{\hat{\mathbb Q}_1}[f_N(\mathbb S^{(1)})](1+L(\exp(C^{(1)}_J |\tilde\kappa_1-\tilde\kappa_2|)-1))]}{1-L(\exp(C^{(1)}_J |\tilde\kappa_1-\tilde\kappa_2|)-1)}\nonumber
\\
& \qquad +\frac{L (\exp(C^{(1)}_J |\tilde\kappa_1-\tilde\kappa_2|)-1)
\mathbb E_{\hat{\mathbb Q}_1}[||\mathbb S^{(1)}||]}
{1-L(\exp(C^{(1)}_J |\tilde\kappa_1-\tilde\kappa_2|)-1)}\nonumber
\\
&\leq\mathbb E_{\hat{\mathbb Q}_1}[f_N(\mathbb S^{(1)})]+C^{(3)}_J|\tilde\kappa_1-\tilde\kappa_2|.\nonumber
\end{align}
Next, we modify the probability measure $\hat{\mathbb Q}_2$ so it will satisfy property (3) of
Definition \ref{dfn.1} for $H^{(2)}_1,...,H^{(2)}_N$.}
Clearly, the measure $\mathbb{Q}\otimes\mathbb Q$ is a probability
measure on $\hat\Omega$,
where the probability measure $\mathbb Q$ is given in Assumption \ref{asm.noarbitrage}.
For any $\lambda\in (0,1)$ consider the probability measure
$$\hat{\mathbb Q}_{\lambda}=\sqrt\lambda [\mathbb{Q}\otimes\mathbb Q]+ (1-\sqrt\lambda)\hat{\mathbb Q}_2.$$
Observe that
$$\mathbb E_{\mathbb{Q}\otimes\mathbb Q}[f_i({\mathbb S}^{(1)})]=\mathbb E_{\mathbb Q}[f_i(\mathbb S)]=D_i, \ \ i\leq N.$$
Set
$\Lambda=\sum_{k=1}^N |H^{(1)}_k-H^{(2)}_k|+|\tilde\kappa_1-\tilde\kappa_2|$.
From \reff{5.3}--\reff{5.newww} and the fact that $D_i< H^{(1)}_i$
 it follows that for $\Lambda$ sufficiently small
\begin{eqnarray*}
&|\mathbb E_{\hat{\mathbb Q}_{\Lambda}}[f_i(\mathbb S^{(1)})]\leq
\sqrt\Lambda D_i+(1-\sqrt\Lambda)(H^{(1)}_i+C^{(3)}_J|\tilde\kappa_1-\tilde\kappa_2|)\leq\\
&H^{(1)}_i-\sqrt\Lambda(H^{(1)}_i-D_i)+C^{(3)}_J\Lambda<H^{(1)}_i-\Lambda\leq H^{(2)}_i.
\end{eqnarray*}
This together with
\reff{5.1+}--\reff{5.1++} yields that $\hat{\mathbb Q}_{\Lambda}\in \mathcal{M}_{\tilde\kappa_2,H^{(2)}}$.
Finally, from \reff{5.1} and \reff{5.2} we obtain
\begin{align*}
\Gamma(H^{(1)},\tilde\kappa_1)-\Gamma(H^{(2)},\tilde\kappa_2)
& \leq\epsilon+\mathbb E_{\hat{\mathbb Q}_1}[G(\mathbb S^{(1)})]-
(1-\sqrt\Lambda)\mathbb E_{\hat{\mathbb Q}_2}[G(\mathbb S^{(1)})]
\\
&\leq\epsilon+C^{(3)}_J|\tilde\kappa_1-\tilde\kappa_2|+\sqrt{\Lambda}K.
\end{align*}
Since $\epsilon>0$ was arbitrary, this completes the proof.
\end{proof}

Now, we are ready to prove the lower bound of Theorem \ref{thm.main}.
\begin{lem}\label{lem6.2}
$$
V^{\mathbb P}_{\kappa}(G)\geq\sup_{\hat{\mathbb Q}\in \mathcal{M}_{\kappa,{\cL}}}\mathbb E_{\hat{\mathbb Q}}[G(\mathbb S^{(1)})].
$$
\end{lem}
\begin{proof}
In view of Lemma \ref{lem5.1},
it is sufficient to prove that
\begin{equation}\label{5.7}
V^{\mathbb P}_{\kappa}(G)\geq \mathbb E_{\hat{\mathbb Q}}[G(\mathbb S^{(1)})],
\end{equation}
for every $\hat{\mathbb Q}\in \mathcal{M}_{\tilde\kappa,\tilde \cL}$
with  $\tilde\kappa<\kappa$ and
$\tilde{\mathcal L}_i<\mathcal{L}_i$, $i\leq N$.

We proceed in two steps. In the first step, we modify the process
$\mathbb S^{(1)}$. In the second step, we apply Lemma \ref{lem3.1}
to the modified process.

\textbf{First step:}
Let $\epsilon>0$.
Define the stopping times,
$\tau^{(\epsilon)}_0:=\tau^{(\epsilon)}_0({\mathbb S^{(1)}})=0$
and for  $k>0$,
$$
\tau^{(\epsilon)}_k:=\tau^{(\epsilon)}_k({\mathbb S^{(1)}})=T
\wedge\inf\left\{t>\tau^{(\epsilon)}_{k-1}:
\mathbb S^{(1)}_t=exp(\pm \epsilon) \mathbb S^{(1)}_{\tau^{(\epsilon)}_{k-1}}
\right\},
$$
and the random variable
$\mathbb K:=\min\{k:\tau^{(\epsilon)}_k=T\}-1<\infty$.
Let $n\in\mathbb N$.
Introduce the stochastic process
$$
\tilde{S}^{(n)}_t=\sum_{i=0}^{n-1}
\mathbb S^{(1)}_{\tau^{(\epsilon)}_i}
{\chi}_{[\tau^{(\epsilon)}_i,\tau^{(\epsilon)}_{i+1})}(t)
+\mathbb S^{(1)}_{\tau^{(\epsilon)}_{\mathbb K\wedge n}}
{\chi}_{[\tau^{(\epsilon)}_n,T]}(t),
\ \ t\in [0,T].
$$
{The stochastic process $\tilde S^{(n)}$ is a pure jump process which agrees with $\mathbb S^{(1)}$
at the jump times $\tau^{(\epsilon)}_1,...,\tau^{(\epsilon)}_{n\wedge\mathbb K}$ and
remains constant afterwards.}

{We argue that for sufficiently large $n$ the terms
$\mathbb E_{\hat{\mathbb Q}}|f_i(\tilde S^{(n)})-f_i(\mathbb S^{(1)})|$, $i=1,...,N$
and $\mathbb E_{\hat{\mathbb Q}}|G(\tilde S^{(n)})-G(\mathbb S^{(1)})|$
are small.
Indeed, as before the fact $\hat{\mathbb Q}\in \mathcal{M}_{\tilde\kappa,\tilde \cL}$
implies that $\mathbb E_{\hat{\mathbb Q}}[\|\mathbb S^{(1)})\|]\leq \hat C$ (where, recall the constant $\hat C$
from Definition \ref{d.measures})
 and so
$\lim_{n\rightarrow\infty} \mathbb E_{\hat{\mathbb Q}}[\|\mathbb S^{(1)}\|{\chi}_{\{\mathbb K\geq n\}}]=0.$
From Assumption \ref{asm2.1} we get
\begin{align*}
\lim\sup_{n\rightarrow\infty}
\left|\mathbb E_{\hat{\mathbb Q}}[f_i(\mathbb S^{(1)})]
-\mathbb E_{\hat{\mathbb Q}}[f_i(\tilde {S}^{(n)})]
\right| & \leq
 \lim\sup_{n\rightarrow\infty}\mathbb E_{\hat{\mathbb Q}}[|f_i(\mathbb S^{(1)})-f_i(\tilde {S}^{(n)})|{\chi}_{\{\mathbb K<n\}}]\\
 & \ \ +
2L \lim_{n\rightarrow\infty}\mathbb E_{\hat{\mathbb Q}}[\|\mathbb S^{(1)}\|{\chi}_{\{\mathbb K\geq n\}}]\\
&\leq L(e^{\epsilon}-1)\mathbb E_{\hat{\mathbb Q}}[\|\mathbb S^{(1)})\|]
\\ & \leq L(e^{\epsilon}-1)\hat C.
\end{align*}
Similarly,
\begin{equation}
\lim\sup_{n\rightarrow\infty}
\left|\mathbb E_{\hat{\mathbb Q}}[G(\mathbb S^{(1)})]
-\mathbb E_{\hat{\mathbb Q}}[G(\tilde {S}^{(n)})]
\right|\leq L(e^{\epsilon}-1)\hat C.
\end{equation}}
{It remains to treat the case $i=N$. From
Assumption \ref{asm.regularity} it follows that there exists $\delta>0$ such
that
$$|\ln x-\ln y|<\delta\Rightarrow q(y)<2(q(x)+x).$$ We conclude
that there exists a constant $C_4$ such that
or any $x,y>0$ we have
$$(1-\tilde\kappa) x \leq y\leq \frac{1}{1-\tilde\kappa} x\Rightarrow
q(y)\leq C_4 (q(x)+x).
$$
This together with
property (2) of Definition \ref{dfn.1}
yields
$$
\mathbb E_{\hat{\mathbb Q}} [q(\mathbb S^{(1)}_{\tau^{(\epsilon)}_n}){\chi}_{\{\mathbb K\geq n\}}]
\leq C_4\mathbb E_{\hat{\mathbb Q}}
\left[ \left(q(\mathbb S^{(2)}_{\tau^{(\epsilon)}_n})+\mathbb S^{(2)}_{\tau^{(\epsilon)}_n}\right){\chi}_{\{\mathbb K\geq n\}}\right].
$$
Since $\mathbb S^{(2)}$ is a martingale and
$\{\mathbb K\geq n\}=\{\tau^{(\epsilon)}_{n}<T\}
\in\hat{\mathbb{F}}_{\tau^{(\epsilon)}_n}$,
then from the Jensen inequality (for the convex function $q(x)+x$) we obtain,}
\begin{align*}
\mathbb E_{\hat{\mathbb Q}} [q(\mathbb S^{(1)}_{\tau^{(\epsilon)}_n}){\chi}_{\{\mathbb K\geq n\}}]
& \leq C_4\mathbb E_{\hat{\mathbb Q}}
\left[ \left(q(\mathbb S^{(2)}_{T})+\mathbb S^{(2)}_{T}\right){\chi}_{\{\mathbb K\geq n\}}\right]\\
& \leq
C_4
\mathbb E_{\hat{\mathbb Q}}
\left[ \left(C_4 q(\mathbb S^{(1)}_{T})+(1+\tilde\kappa)\mathbb S^{(1)}_{T}\right){\chi}_{\{\mathbb K\geq n\}}\right].
\end{align*}
Thus the inequality $\mathbb E_{\hat{\mathbb Q}}[(q(\mathbb S^{(2)}_{T})]<\infty$ implies
\begin{align*}
 \lim\sup_{n\rightarrow\infty}
&\left|\mathbb E_{\hat{\mathbb Q}}[f_N(\mathbb S^{(1)})]
-\mathbb E_{\hat{\mathbb Q}}[f_N(\tilde {S}^{(n)})]
\right|  \\ &\qquad \leq
\lim\sup_{n\rightarrow\infty}\mathbb E_{\hat{\mathbb Q}}\left[\left(f_N(\mathbb S^{(1)})+
f_N(\tilde {S}^{(n)})\right){\chi}_{\{\mathbb K\geq n\}}\right]
\\ &\qquad \leq
\lim\sup_{n\rightarrow\infty}C_4
\mathbb E_{\hat{\mathbb Q}}
\left[(C_4+1/C_4) \left(q(\mathbb S^{(1)}_{T})
+(1+\tilde\kappa)\mathbb S^{(1)}_{T}\right){\chi}_{\{\mathbb K\geq n\}}\right]\\
&\qquad  =0.
\end{align*}
We conclude that for sufficiently large $n$
\begin{eqnarray}\label{5.12}
&\left|\mathbb E_{\hat{\mathbb Q}}[G(\mathbb S^{(1)})]
-\mathbb E_{\hat{\mathbb Q}}[G(\tilde{ S}^{(n)})]
\right|\leq 2L(e^{\epsilon}-1)\hat C \ \ \mbox{and}\\
&\left|\mathbb E_{\hat{\mathbb Q}}[f_i(\mathbb S^{(1)})]
-\mathbb E_{\hat{\mathbb Q}}[f_i(\tilde {S}^{(n)})]
\right|\leq 2L(e^{\epsilon}-1)\hat C \ \ i\leq N.
\nonumber
\end{eqnarray}
We fix $n$ sufficiently large that
the above inequalities hold and set $\tilde{S}:=\tilde{S}^{(n)}$.

{Next, we modify the jump times so they will lie on a grid.}
Let $m\in\mathbb N$.
Define by recursion the following sequence of random variables,
\begin{eqnarray*}
\hat\tau^{(\epsilon)}_k&:=&\sum_{i=1}^k\ \Delta \hat \tau^{(\epsilon)}_i, \ \ \mbox{where}\\
\Delta  \hat\tau^{(\epsilon)}_i&=& \min\{\Delta t \in
\{T/m,2 T/m,...,T\}: \Delta t\geq \Delta \tau^{(\epsilon)}_i:= \tau^{(\epsilon)}_i-\tau^{(\epsilon)}_{i-1}\},
\end{eqnarray*}
and
$$\sigma_k=T{\chi}_{\{\tau^{(\epsilon)}_k=T\}}+\hat \tau^{(\epsilon)}_k\wedge(T(1-2^{-k}/m))
{\chi}_{\{\tau^{(\epsilon)}_k<T\}}, \ \ k=0,1,...,n.
$$
Observe that for any $i$, $\sigma_{i+1}\geq \sigma_i$
and $\sigma_{i+1}=\sigma_i$ if and only if $\sigma_i=T$.
Notice that $\sigma_1,...,\sigma_{n}$ are not (in general) stopping times
with respect to the filtration $\hat{\mathbb F}$. Define the stochastic
process
$$
\dot{S}_t:=\dot{S}^{(m)}_t
=\sum_{i=0}^{n-1}\mathbb S^{(1)}_{\tau^{(\epsilon)}_i}
{\chi}_{[\sigma_i,\sigma_{i+1})}(t)
+\mathbb S^{(1)}_{\tau^{(\epsilon)}_{\mathbb K\wedge n}}
{\chi}_{[\sigma_n,T]}(t), \ \ t\in [0,T].
$$

\textbf{Second step:}
The process $\dot{S}_t$ is a piecewise constant process, and the jump times
are lying on a finite grid. Thus the natural filtration
which is generated by $\dot S$ is right continuous, and so
the martingale
$$\hat M_t:=\mathbb E_{\hat{\mathbb Q}}(\mathbb S^{(2)}_T|\dot S_u, u\leq t)$$
is a $c\grave{a}dl\grave{a}g$ martingale.
Let $k\leq n$. Clearly, $\sigma_k$ is a stopping time with respect to the natural filtration
generated by $\dot S$. Furthermore $\dot S_{[0,\sigma_k]}$ is measurable with respect to
$\hat{\mathbb F}_{\tau^{(\epsilon)}_k}$.
This together with the fact that
$$
e^{-\epsilon}\leq\frac{\dot{S}_{\sigma_k}}{\mathbb S^{(1)}_{\tau^{(\epsilon)}_k}}
\leq e^{\epsilon}
$$
and properties (1)--(2) in Definition \ref{dfn.1}, imply that
\begin{eqnarray*}
|\hat M_{\sigma_k}-\dot S_{\sigma_k}|&=&
\left|\mathbb E_{\hat{\mathbb Q}}\left(
\mathbb E_{\hat{\mathbb Q}}[\mathbb S^{(2)}_T\ |\
\hat{\mathbb F}_{\tau^{(\epsilon)}_k}]\ \left|\right.
\dot S_u, u\leq \sigma_k\right)-\dot S_{\sigma_k}\right|\\
&\leq&
\dot S_{\sigma_k}((1+\tilde\kappa)e^{\epsilon}-1)
\leq\dot S_{\sigma_k}(\tilde\kappa+2\epsilon),
\end{eqnarray*}
where in the last equality we assume that $\epsilon$ is sufficiently small.
Let $\sigma_{n+1}=T$.
Then, for any $k\leq n$ and $t\in [\sigma_k,\sigma_{k+1}]$,
we conclude that
$$
e^{-2\epsilon}(1-\tilde\kappa-2\epsilon)\dot S_t
\leq\hat M_{\sigma_{k+1}}\leq e^{2\epsilon}(1+\tilde\kappa+2\epsilon)\dot S_t.
$$
Since $\hat M$ is a martingale with respect to the natural filtration
of $\dot S$, we conclude that for sufficiently small $\epsilon$,
\begin{equation}\label{5.20}
|\hat M_t-\dot S_t|\leq (1+\tilde\kappa+5\epsilon)\dot S_t.
\end{equation}
Clearly,
$$
\lim_{m\rightarrow\infty}
\|\tilde S-\dot S^{(m)}\|=0, \ \ \hat{\mathbb Q} \ \ \mbox{a.s.}
$$
Observe that the above processes are uniformly bounded.
Hence, by
Assumptions \ref{asm2.1}--\ref{asm.regularity},
\begin{eqnarray}\label{5.10}
\mathbb E_{\hat{\mathbb Q}}[G(\tilde S)]&=&\lim_{m\rightarrow\infty}
\mathbb E_{\hat{\mathbb Q}}[G(\dot S^{(m)})] \ \ \mbox{and} \\
\mathbb E_{\hat{\mathbb Q}}[f_i(\tilde S)]&=&\lim_{m\rightarrow\infty}
\mathbb E_{\hat{\mathbb Q}}[f_i(\dot S^{(m)})], \ \ i\leq N.\nonumber
\end{eqnarray}
Denote by $\dot{\mathbb Q}_m$ the distribution
of $\dot S^{(m)}$ on the space $\mathbb D[0,T]$.
Let us choose $\epsilon$ such that
$\hat\kappa:=\tilde\kappa+6\epsilon$ is satisfies
$$\min\left(\frac{1+\kappa}{1+\hat\kappa},\frac{1-\hat\kappa}{1-\kappa}\right)\geq e^{2\epsilon},$$
and
\begin{align*}
\mathcal L_i-L (\hat C+\mathcal L_N)(e^{4\epsilon}+\epsilon-1)&>3 L(e^{\epsilon}-1)\hat C+\tilde{\mathcal L_i}, \ \ i<N,\\
\frac{\cL_N(1-L(e^{\epsilon}-1))-L\hat C(e^{\epsilon}-1)}{1+L(e^\epsilon-1)}&
>3 L(e^{\epsilon}-1)\hat C+\tilde{\mathcal L}_N.
\end{align*}
From \reff{5.12}--\reff{5.10}, it follows that for sufficiently large $m$
the measure
$\dot{\mathbb Q}_m\in \mathcal{M}^{\cT,\epsilon}_{\hat\kappa,{\cL}}$
with the choice $\cT:=\{k T 2^{-n}/m\}_{k=0}^{2^n m}$. Thus, in
view of Lemma \ref{lem3.1}, we have
$$
V^{\mathbb P}_{\kappa}(G)
\geq \mathbb E_{\hat{\mathbb Q}}[G(\dot S^{(m)})]-L\hat C (e^{4\epsilon}+\epsilon-1).
$$
We now apply \reff{5.12}, \reff{5.10} and take
the limit as $m$ tends to infinity.
The result is
$$
V^{\mathbb P}_{\kappa}(G)
\geq \mathbb E_{\hat{\mathbb Q}}[G(\mathbb {S}^{(1)})]-2 L(e^{\epsilon}-1)\hat C
-L\hat C (e^{4\epsilon}+\epsilon-1).
$$
Now, \reff{5.7} follows after taking the limit as $\epsilon$ tends to zero.
\end{proof}

Next, we establish the upper bound (\ref{e.upper}).
\begin{lem}\label{lem6.3}
$$
V_{\kappa}(G)
\leq\sup_{\hat{\mathbb Q}\in \mathcal{M}_{\kappa,{\cL}}}
\mathbb E_{\hat{\mathbb Q}}[G(\mathbb S^{(1)})].
$$
\end{lem}
\begin{proof}
Let  $\mathbb Q$ be the probability measure from Assumption
\ref{asm.noarbitrage}.  Then,
$\mathbb Q\otimes\mathbb Q
\in \mathcal{M}_{\kappa,({\mathcal L}_1,...,{\mathcal L}_N)}$.
Therefore, if $V_{\kappa}(G)\leq 0$,
then \reff{e.upper} is trivial.  So we may assume without loss of generality
that $V_{\kappa}(G)>0$.
Choose  $\epsilon>0$, $\Lambda>1$, $\hat\kappa>\tilde\kappa>\kappa$ and
$\tilde{\mathcal L}_i>\mathcal{L}_i$, $i\leq N$.
Assume that $\epsilon$ is sufficiently small so
$L(e^{2\epsilon}+\epsilon-1)
\frac{\hat C^2}{2(1-8\kappa)}<V_{\kappa}(G)$
and
$\tilde\kappa$  satisfies \reff{condition1}.
This together with Lemma \ref{lem4.1}
yields that
there exists a probability measure
$\tilde{\mathbb Q}\in \mathcal{M}^{\epsilon,\Lambda}_{\tilde\kappa,{\cL}}$
such that
\begin{equation}
\label{5.15}
V_{\kappa}(G)<\mathbb E_{\tilde{\mathbb Q}}[G(\tilde{\mathbb S})]+L(e^{2\epsilon}+\epsilon-1)
\frac{\hat C^2}{(1-8\kappa)}.
\end{equation}
Next, we proceed in three steps.
In the first step (similarly to Lemma \ref{lem6.2}), we  modify the stochastic
process $\tilde{\mathbb S}$. In the second step,
 we use the Wiener  space in order to construct a continuous
consistent price system with (almost) the required properties.
In the last step, we modify again the constructed
continuous consistent price system in order to get rid of the truncation in the term
$f_N(\mathbb S^{(1)})\wedge\Lambda\mathbb{S}^{(1)}_T$.
Finally, we Apply Lemma \ref{lem5.1}.

\textbf{First step:}
Let
$$(1-\tilde\kappa)\tilde{\mathbb S}_t\leq\tilde {M}_t\leq (1+\tilde\kappa)\tilde{\mathbb S}_t, \ \ t\in [0,T],$$
be the associated martingale {corresponding to the probability measure
$\tilde{\mathbb Q}\in \mathcal{M}^{\epsilon,\Lambda}_{\tilde\kappa,{\cL}}$}.
Let
$\tilde\tau^{(\epsilon)}_0:=\tilde\tau^{(\epsilon)}_0(\tilde{\mathbb S})=0$,
and for $k>0$ set,
$$
\tilde\tau^{(\epsilon)}_k:=\tilde\tau^{(\epsilon)}_k(\tilde{\mathbb S})=T\wedge
\inf\left\{t>\tilde\tau^{(\epsilon)}_{k-1}:
|\ln {\tilde{\mathbb S}}_{\tilde\tau^{(\epsilon)}_{k+1}}-\ln {\tilde{\mathbb S}}_{\tilde\tau^{(\epsilon)}_{k}}|=\epsilon\right\}
$$
and
$\tilde {\mathbb K}=\min\{k:\tilde\tau^{(\epsilon)}_k=T\}-1<\infty$.
{Observe that the probability measure $\tilde{\mathbb Q}$ supported on
$\mathbb D^{(\epsilon)}$ and so $\tilde\tau_k$, $k\geq 0$ are indeed stopping times.}

Let $n\in\mathbb N$.
Set,
$$
\tilde{S}^{(n)}_t:=\sum_{i=0}^{n-1}
\tilde{\mathbb S}_{\tilde\tau^{(\epsilon)}_i}
{\chi}_{[\tilde\tau^{(\epsilon)}_i,\tilde\tau^{(\epsilon)}_{i+1})}(t)
+\tilde{\mathbb S}_{\tilde\tau^{(\epsilon)}_{\tilde{\mathbb K}\wedge n}}
{\chi}_{[\tilde\tau^{(\epsilon)}_n,T]}(t), \ \ t\in [0,T].
$$
{From the definition
of the set $\mathcal{M}^{\epsilon,\Lambda}_{\tilde\kappa,{\cL}}$
it follows that $\mathbb{E}_{\tilde{\mathbb Q}}[q({\tilde{\mathbb S}}_T)
\wedge\Lambda(\tilde{\mathbb S}_T+1)]<\infty$, and so
and $\mathbb{E}_{\tilde{\mathbb Q}}[\tilde{\mathbb S}_T]<\infty$, as well.}
{Moreover,}
\begin{eqnarray*}
\mathbb E_{\tilde{\mathbb Q}} [\tilde{\mathbb S}_{\tilde\tau^{(\epsilon)}_n}{\chi}_{\{\tilde {\mathbb K}\geq n\}}]
&\leq&
(1+\tilde\kappa)\mathbb E_{\tilde{\mathbb Q}} [\tilde M_{\tilde\tau^{(\epsilon)}_n}
{\chi}_{\{\tilde {\mathbb K}\geq n\}}]=
(1+\tilde\kappa)\mathbb E_{\tilde{\mathbb Q}} [\tilde M_T{\chi}_{\{\tilde {\mathbb K}\geq n\}}]\\
&\leq&
(1+\tilde\kappa)^2\mathbb E_{\tilde{\mathbb Q}} [\tilde{\mathbb S}_T{\chi}_{\{\tilde {\mathbb K}\geq n\}}].
\end{eqnarray*}
{We conclude that}
\begin{equation}\label{5.nnn}
\lim_{n\rightarrow\infty} \mathbb E_{\tilde{\mathbb Q}} [(\tilde{\mathbb S}_{\tilde\tau^{(\epsilon)}_n}+\tilde{\mathbb S}_T){\chi}_{\{\tilde {\mathbb K}\geq n\}}]=0.
\end{equation}
{As in the proof of Lemma \ref{lem4.1}, we will use the fact that
$f_i(S)$, $i<N$ are bounded (from both sides) by a multiply of
$1+S_T$. This together with (\ref{5.nnn}) and the fact that $\tilde S^{(n)}=\tilde {\mathbb S}$ on the event
$\{n>\tilde{\mathbb K}\}$
 yields that for sufficiently large $n$,}
\begin{eqnarray}
\label{5.12+}
&&
\left|\mathbb E_{\tilde{\mathbb Q}}[G(\tilde{\mathbb S})]
-\mathbb E_{\tilde{\mathbb Q}}[G(\tilde{S}^{(n)})]\right|\leq \epsilon,\\
&&
\left|\mathbb E_{\tilde{\mathbb Q}}[f_i(\tilde{\mathbb S})]
-\mathbb E_{\tilde{\mathbb Q}}[f_i(\tilde {S}^{(n)})]\right|\leq \epsilon, \ \ i\leq N-1,
\nonumber\\
&&
\left|\mathbb{E}_{\tilde{\mathbb Q}}[q({\tilde{\mathbb S}}_T)
\wedge\Lambda(\tilde{\mathbb S}_T+1)]-
\mathbb{E}_{\tilde{\mathbb Q}}[q(\tilde S^{(n)}_T)\wedge\Lambda (\tilde{S}^{(n)}_T+1)]
\right|
\leq\epsilon.\nonumber
\end{eqnarray}
We choose $n$ sufficiently
large and set $\tilde{S}:=\tilde{S}^{(n)}$.

Next, let $m\in\mathbb N$.
Define by recursion the following sequence of random variables,
\begin{eqnarray*}
\hat\tau^{(\epsilon)}_k&:=&\sum_{i=1}^k\ \Delta \hat \tau^{(\epsilon)}_i, \ \ \mbox{where}\\
\Delta  \hat\tau^{(\epsilon)}_i&=& \min\{\Delta t \in
\{T/m,2 T/m,...,T\}: \Delta t\geq \Delta\tilde \tau^{(\epsilon)}_i:= \tilde\tau^{(\epsilon)}_i-\tilde\tau^{(\epsilon)}_{i-1}\},
\end{eqnarray*}
and
$$
\sigma_k=T{\chi}_{\{\tilde\tau^{(\epsilon)}_k=T\}}+\hat \tau^{(\epsilon)}_k\wedge(T(1-2^{-k}/m))
{\chi}_{\{\tilde\tau^{(\epsilon)}_k<T\}}, \ \ k=0,1,...,n.
$$
Similarly, to Lemma \ref{lem6.2} we have that for any $i$, $\sigma_{i+1}\geq \sigma_i$
and $\sigma_{i+1}=\sigma_i$ if and only if $\sigma_i=T$.
Define the stochastic
process
$$
\dot{S}_t:=\dot{S}^{(m)}_t=\sum_{i=0}^{n-1}\tilde{\mathbb S}_{\tilde\tau^{(\epsilon)}_i}
{\chi}_{[\sigma_i,\sigma_{i+1})}(t)
+\tilde{\mathbb  S}_{\tilde\tau^{(\epsilon)}_{\tilde{\mathbb K}\wedge n}}{\chi}_{[\sigma_n,T]}(t),
\ \ t\in [0,T].
$$
Again, as in Lemma \ref{lem6.2} the process $\dot{S}_t$ is a piecewise constant process, and the jump times
are lying on a finite grid. Introduce the ($c\grave{a}dl\grave{a}g$) martingale
$$
\hat M_t:=\mathbb E_{\hat{\mathbb Q}}(\tilde M_T|\dot S_u, u\leq t).
$$

By using the same arguments as in (\ref{5.20})--(\ref{5.10}) we get
\begin{equation}
\label{5.200}
|\hat M_t-\dot S_t|\leq (1+\tilde\kappa+5\epsilon)\dot S_t,
\end{equation}
and
\begin{eqnarray}\label{5.100}
\mathbb E_{\tilde{\mathbb Q}}[G(\tilde S)]&=&\lim_{m\rightarrow\infty}
\mathbb E_{\tilde{\mathbb Q}}[G(\dot S^{(m)})], \\
\mathbb E_{\tilde{\mathbb Q}}[f_i(\tilde S)]&=&\lim_{m\rightarrow\infty}
\mathbb E_{\tilde{\mathbb Q}}[f_i(\dot S^{(m)})], \ \ i\leq N-1 \nonumber\\
\mathbb E_{\tilde{\mathbb Q}}[q(\tilde S_T)\wedge\Lambda(\tilde {\mathbb S}_T+1)]
&=&\lim_{m\rightarrow\infty}
\mathbb E_{\tilde{\mathbb Q}}[q(\dot S^{(m)}_T)\wedge\Lambda(\dot {S}^{(m)}_T+1)].
\nonumber
\end{eqnarray}
From \reff{5.12+} and \reff{5.100},
 it follows that we can choose
$m$ sufficiently large such that
\begin{eqnarray}\label{5.17}
&&
\left|\mathbb E_{\tilde{\mathbb Q}}[G(\tilde{\mathbb S})]
-\mathbb E_{\tilde{\mathbb Q}}[G(\dot S^{(m)})]\right|\leq 2\epsilon,\\
&&
\left|\mathbb E_{\tilde{\mathbb Q}}[f_i(\tilde{\mathbb S})]
-\mathbb E_{\tilde{\mathbb Q}}[f_i(\dot S^{(m)})]\right|
\leq 2\epsilon, \ \ i\leq N-1,\nonumber\\
&&
\left|\mathbb{E}_{\tilde{\mathbb Q}}[q({\tilde{\mathbb S}}_T)
\wedge\Lambda(\tilde{\mathbb S}_T+1)]-
\mathbb{E}_{\tilde{\mathbb Q}}[q(\dot S^{(m)}_T)
\wedge\Lambda(\dot {S}^{(m)}_T+1)]\right|
\leq2\epsilon.\nonumber
\end{eqnarray}
Choose such $m$ and denote $\dot{S}=\dot{S}^{(m)}$. The stochastic
process ${\{\dot S_t\}}_{t=0}^T$ is a piecewise constant process,
and the jump times
are lying on a finite grid. Denote the grid
by $\cT=\{t_1,...,t_r,T\}$, where
$0=t_0<t_1<...<t_r<T$.

\textbf{Second step:}
Let
$(\Omega^W,\mathcal{F}^W,\mathbb{P}^W)$ be a complete probability
space
together with a standard
Brownian motion
and the natural filtration
$\mathcal{F}^W_t=\sigma{\{W_s|s\leq{t}\}}$.

From Theorem 1 in Skorokhod (1976) and the fact that
the random variables $W_{t_{i+1}}-W_{t_i}$, $i=0,..,r-1$ are independent,
it follows that
we can find a sequence of measurable function
$g^{(1)}_i,g^{(2)}_i:\mathbb{R}^{2i-1}\rightarrow\mathbb{R}$, $i=1,...,r$ with the following property.
The stochastic processes (adapted to the Brownian filtration)
${\{\dot{S}^W_{t_i}\}}_{i=0}^r$
and ${\{\hat{M}^W_{t_i}\}}_{i=0}^r$
which are given by
the recursion relations
$$\dot{S}^W_{t_0}=1, \ \ \hat{M}^W_{t_0}=\hat M_0$$
and for $i>0$
\begin{eqnarray*}
&\dot {S}^W_{t_i}=g^{(1)}_i(W_{t_{i+1}}-W_{t_i},\dot{S}^W_{t_0},...,\dot{S}^W_{t_{i-1}},
\hat{M}^W_{t_0},...,\hat{M}^W_{t_{i-1}}), \\
&\hat {M}^W_{t_i}=g^{(2)}_i(W_{t_{i+1}}-W_{t_i},\dot{S}^W_{t_0},...,\dot{S}^W_{t_{i-1}},
\hat{M}^W_{t_0},...,\hat{M}^W_{t_{i-1}})
\end{eqnarray*}
have the same joint distribution as the processes
${\{\dot{S}_{t_i}\}}_{i=0}^r$
and ${\{\hat{M}_{t_i}\}}_{i=0}^r$.
Namely, the distribution of
$$(\dot {S}^W_{t_0},...,\dot {S}^W_{t_r},\hat{M}^W_{t_0},...,\hat{M}^W_{t_r})$$
under the probability measure $\mathbb{P}^W$ is equals to
the distribution of
$$(\dot {S}_{t_0},...,\dot {S}_{t_r},\hat{M}_{t_0},...,\hat{M}_{t_r})$$
under the probability measure $\tilde{\mathbb Q}$.

Since the Brownian motion increments are independent,
 for any $i<r$,
$$
\mathbb E_{\mathbb P^W}(\hat M^W_{t_{i+1}}|\mathcal F^W_{t_i})=
\mathbb E_{\mathbb P^W}(\hat M^W_{t_{i+1}}|
\dot S^W_{t_1},...,\dot S^W_{t_i},\hat M^W_{t_1},...,\hat M^W_{t_i})=\hat M^W_{t_i}.
$$
Thus, we
can extend the martingale
${\{\hat{M}^W_{t_i}\}}_{i=0}^r$
to a continuous time martingale (Brownian martingale)
$$
\hat M^W_t=\mathbb E_{\mathbb P^W}(\hat M^W_{t_r}|\mathcal F^W_t), \ \ t\in [0,T].
$$

Next, we define the stochastic process ${\{S^W_t\}}_{t=0}^T$
by the following linear interpolation,
$$
S^W_t={\chi}_{[0,t_1]}(t)+\sum_{i=1}^r
\frac{(t-t_i)\dot S^W_{t_i}+(t_{i+1}-t)\dot S^W_{t_{i-1}}}{t_{i+1}-t_i}\
{\chi}_{(t_i,t_{i+1}]}(t),
$$
where we set $t_{r+1}=T$. Observe that the stochastic process
$S^W$ is continuous and adapted to the Brownian filtration.
Since
$$
\frac{{\dot S}^W_{t_{i+1}}}{\dot S^W_{t_i}}\in\{1,e^{\epsilon},e^{-\epsilon}\},
$$
it follows from \reff{5.200} that (for $\epsilon$ sufficiently small)
\begin{equation}\label{5.18}
\left|\hat M^W_t-{S}^W_t\right|
\leq (\tilde\kappa+10\epsilon){S}^W_t, \ \ t\in [0,T].
\end{equation}
Set,
$$
\dot S^W_t=\sum_{i=0}^{r-1}\dot S^W_{t_i}{\chi}_{[t_i,t_{i+1})}(t)+
\dot S^W_{t_r}{\chi}_{[t_r,T]}(t), \ \ t\in [0,T].
$$
Clearly, the processes $\dot S^W$ and $\dot S$ have the same distribution and
and consequently,
\begin{eqnarray}
\label{5.19}
\mathbb E_{\mathbb P^W}[G(\dot S^W)]&=&
\mathbb E_{\tilde{\mathbb Q}}[G(\dot S)],\\
\mathbb E_{\mathbb P^W}[f_i(\dot S^W)]&=&
\mathbb E_{\tilde{\mathbb Q}}[f_i(\dot S)], \ \ i\leq N-1,\nonumber\\
\mathbb E_{\mathbb P^W}[q(\dot S^W_T)\wedge\Lambda(\dot S^W_T+1)]&=&
\mathbb{E}_{\tilde{\mathbb Q}}[q(\dot S_T)\wedge\Lambda(\dot {S}_T+1)].\nonumber
\end{eqnarray}
Also, \reff{5.17} and \reff{5.19} imply that
$$
\mathbb E_{\mathbb P^W}[q(\dot S^W_T)\wedge\Lambda(\dot {S}^W_T+1)]
\leq 2\epsilon+ {{\mathcal L}}_N+B,
$$
where $B$ is given in Definition \ref{dfn4.1}. Therefore,
 there exists a constant
$C$ (which does not depend on $\epsilon>0$ and $\Lambda>1$) such that
$\mathbb E_{\mathbb P^W}\dot S^W_T\leq C$.
This together with the Kolmogorov inequality for the martingale
$\hat M^W$ yield that
\begin{eqnarray*}
&\mathbb P^W\left(\|S^W\|>\frac{1}{\sqrt\epsilon}\right)\leq
\mathbb P^W\left(\|\hat{M}^W\|>\frac{1}{(1+\tilde\kappa+10\epsilon)\sqrt\epsilon}\right)\leq\\
&\mathbb E_{\mathbb P^W}[\hat M^W_T
(1+\tilde\kappa+10\epsilon)\sqrt\epsilon]
\leq C (1+\tilde\kappa+10\epsilon)^2  \sqrt\epsilon.
\end{eqnarray*}
Observe that by construction $\|S^W-\dot S^W\|\leq 4\epsilon ||S^W||$. {Thus from Assumption \ref{asm2.1} it follows that}
\begin{eqnarray*}
&\mathbb E_{\mathbb P^W}[|G(S^W)-G(\dot S^W)|]\leq \mathbb E_{\mathbb P^W}[K \chi_{\{|S^W\|>1/\sqrt \epsilon\}}+4L\sqrt\epsilon
\chi_{\{|S^W\|\leq 1/\sqrt \epsilon\}}]\\
&\leq ( K C(1+\tilde\kappa+10\epsilon)^2 + 4L)\sqrt\epsilon.
\end{eqnarray*}
{Similarly for path dependent $f_i$ we have}
$$\mathbb E_{\mathbb P^W}[|f_i(S^W)-f_i(\dot S^W)|]\leq (2 \|f_i\|_{\infty} C(1+\tilde\kappa+10\epsilon)^2 + 4L)\sqrt\epsilon$$
{where $\|f_i\|_{\infty} $ is the uniform bound of the path dependent claim $|f_i|$.}
{Since $S^W_T=\dot S^W_T$ then for non path dependent $f_i$ we have a trivial estimate.}
We now use these inequalities together with
\reff{5.17} and \reff{5.19}, to construct
 a constant $\tilde C$ satisfying,
\begin{eqnarray}\label{5.20+}
&&
\left|\mathbb E_{\mathbb P^W}[G(S^W)]
-\mathbb E_{\tilde{\mathbb Q}}[G(\tilde{\mathbb S})]\right|
\leq \tilde C\sqrt\epsilon,\\
&&
\left|\mathbb E_{\mathbb P^W}[f_i(S^W)]
-\mathbb E_{\tilde{\mathbb Q}}[f_i(\tilde{\mathbb S})]\right|
\leq \tilde C\sqrt\epsilon, \ \ i\leq N-1,\nonumber\\
&&
\left|\mathbb{E}_{\tilde{\mathbb Q}}
\left[q(S^W_T)\wedge\Lambda ({S}^W_T+1)\right]
-\mathbb{E}_{\tilde{\mathbb Q}}
\left[q({\tilde{\mathbb S}}_T)\wedge\Lambda(\tilde{\mathbb S}_T+1)\right]
\right |\leq \tilde C\sqrt\epsilon.\nonumber
\end{eqnarray}

\textbf{Third step:}
Let $x_{\Lambda}$ be the solution of the equation
$q(x)=\Lambda (x+1)$ where we assume that $\Lambda>q(0)$
so the equation has exactly one solution. {Indeed (if by contradiction) we have two solutions
$0<x<y$ then
$$\frac{q(y)-q(x)}{y-x}=\Lambda<\frac{q(x)-q(0)}{x}$$
and we get contradiction to convexity}.
Define the stochastic processes by,
$$
\rho_t:=\frac{\hat M^W_t}{S^W_t}, \qquad
M_t:=\mathbb E_{\mathbb P^W}(\hat M^W_T\wedge \rho_T x_{\Lambda}|\mathcal F^W_t),
$$
and
$$
S_t:=\frac{M_t }{\rho_t}\frac{t+(T-t)\rho_0/M_0}{T}, \ \ t\in[0,T].
$$
In view of \reff{5.20+},
\begin{align}\label{5.20++}
\mathbb {E}_{\mathbb P^W}
\left[\hat M^W_T{\chi}_{\hat M^W_T>\rho_T x_{\Lambda}}\right]
& \leq
2\mathbb {E}_{\mathbb P^W}\left[S^W_T{\chi}_{ S^W_T>x_{\Lambda}}\right]\\
&\leq \frac{2}{\Lambda}\mathbb{E}_{\tilde{\mathbb Q}}
\left[q(S^W_T)\wedge\Lambda ({S}^W_T+1)\right]\nonumber
\\ & \leq
\frac{2 (\tilde C\sqrt\epsilon+\mathcal L_N+B)}{\Lambda}.\nonumber
\end{align}
Thus $|M_0-\rho_0|=|M_0-\hat M^W_0|\leq{C_1}/{\Lambda}$
for some constant $C_1$.
This together with (\ref{5.18}) implies
that for sufficiently large $\Lambda$
we have the following inequality,
\begin{equation}\label{5.21}
| M_t-{S}_t|\leq \left(\tilde\kappa+10\epsilon+\frac{1}{\sqrt\Lambda}\right){S}_t, \ \ t\in [0,T].
\end{equation}
Next, consider the martingale
$$
m_t:=\mathbb {E}_{\mathbb P^W}\left[
\hat M^W_T{\chi}_{\{\hat M^W_T>\rho_T x_{\Lambda}\}}\ |
\mathcal F^W_t\right], \ \ t\in [0,T].
$$
{Observe that $0\leq\hat M^W_t-M_t\leq m_t$, $t\in [0,T]$.
Thus we obtain that there exists a constant $C_2$ such that}
\begin{align}\label{5.cal}
\|S^W-S\| & \leq |M^W-M\|\sup_{0\leq t\leq T}\frac{1}{\rho_t}+
\|M\| \sup_{0\leq t\leq T}\left|\frac{1}{\rho_t}-\frac{t+(T-t)\rho_0/M_0}{T\rho_t}\right|\\
& \leq
2\|m\|+\frac{C_2}{\Lambda}\|M^W\|.\nonumber
\end{align}
The Kolmogorov inequality and \reff{5.20++} imply that
$$\mathbb P^W\left(\|m\|>{1}/{\sqrt\Lambda}\right)\leq {C_3}/{\sqrt\Lambda}$$
for some constant $C_3$.
Moreover,
$$\mathbb P^W\left(\|M^W\|>{\sqrt\Lambda}\right)\leq \frac{M^W_0}{\sqrt\Lambda}\leq\frac{2}{\sqrt\Lambda}.$$
From (\ref{5.cal}) we conclude that
\begin{equation*}
\mathbb P^W\left(\|S^W-S\|>\frac{2+C_2}{\sqrt\Lambda}\right)\leq \frac{2+C_3}{\sqrt\Lambda}.
\end{equation*}
Thus from Assumption \ref{asm.regularity} it follows that
\begin{eqnarray}\label{5.22}
&\mathbb E_{\mathbb P^W}[|G(S^W)-G(S)|]\leq \\
&\mathbb E_{\mathbb P^W}[K \chi_{\{|S^W-S\|>\frac{2+C_2}{\sqrt\Lambda}\}}+L\frac{2+C_2}{\sqrt\Lambda}
\chi_{\{|S^W\|\leq \frac{2+C_2}{\sqrt\Lambda} \}}]\leq \frac{C_4}{\sqrt\Lambda}\nonumber
\end{eqnarray}
for some constant $C_4$.
{Similarly for path--dependent $f_i$ we get}
\begin{equation}\label{5.22new}
\mathbb E_{\mathbb P^W}[|f_i(S^W)-f_i(S)|]\leq \frac{C_4}{\sqrt\Lambda}.
\end{equation}
{For non path--dependent $f_i$, $i<N$ we have}
\begin{align}\label{5.22neww}
\mathbb E_{\mathbb P^W}[|f_i(S^W)-f_i(S)|]& \leq L\mathbb E_{\mathbb P^W}[|S^W_T-S_T|]\\
& \leq
L\mathbb {E}_{\mathbb P^W}
\left[S^W_T{\chi}_{ S^W_T>x_{\Lambda}}\right]
\nonumber \\
& \leq \frac{L (\tilde C\sqrt\epsilon+\mathcal L_N+B)}{\Lambda}\nonumber
\end{align}
 where the last inequality follows from
(\ref{5.20++}).
The only remanning delicate point
is $i=N$. From the fact that
$S_T=S^W_T\wedge x_{\Lambda}$
we get
$$
\mathbb E_{\mathbb P^W}[q(S_T)]\leq
\mathbb E_{\mathbb P^W}[q(S^W_T)\wedge \Lambda(S^W_T+1)].
$$
This together with (\ref{5.20+}), (\ref{5.21}) and (\ref{5.22})--(\ref{5.22neww}) yields
that for sufficiently large $\Lambda$ and small $\epsilon>0$ the distribution of
$(S,M)$ on the space
$\hat\Omega:=\Omega\times \mathcal{C}^{++}_{[0,T]}$ is an element in
$\mathcal{M}_{\hat\kappa,(\tilde{\mathcal L}_1,...,\tilde{\mathcal L}_N)}.$ Furthermore,
$$
\left|\mathbb E_{\mathbb P^W}[G(S)]-\mathbb E_{\tilde{\mathbb Q}}[G(\tilde{\mathbb S})]
\right|\leq \tilde C\sqrt\epsilon+\frac{C_4}{\sqrt \Lambda}.
$$
We now use \reff{5.15}, to obtain
$$
V_{\kappa}(G)<L(e^{2\epsilon}+\epsilon-1)
\frac{\hat C^2}{(1-8\kappa)}+\tilde C \sqrt\epsilon+\frac{C_4}{\sqrt \Lambda}+
\sup_{\hat{\mathbb Q}\in \mathcal{M}_{\hat\kappa,(\tilde{\mathcal L}_1,...,\tilde{\mathcal L}_N)}}\mathbb E_{\hat{\mathbb Q}}[G(\mathbb S^{(1)})].
$$
Finally we apply Lemma \ref{lem5.1} and take
the limits $\Lambda\rightarrow\infty$, $\epsilon\downarrow 0$,
$\hat\kappa\downarrow\kappa$,
$\tilde{\mathcal{L}}_i\downarrow\mathcal L_i$, $i\leq N$.
The result is
$$
V_{\kappa}(G)\leq
\sup_{\hat{\mathbb Q}\in \mathcal{M}_{\kappa,\cL}}
\mathbb E_{\hat{\mathbb Q}}[G(\mathbb S^{(1)})].
$$
This concludes the proof of the lemma as well as the proof
of the main result.
\end{proof}

\end{document}